\documentclass[a4paper,onecolumn,unpublished]{quantumarticle}
\pdfoutput=1

\usepackage{graphicx}
\usepackage{amsmath}
\usepackage{amsthm}
\usepackage{amssymb}
\usepackage{xspace}
\usepackage{tikzit}
\tikzstyle{gate}=[shape=rectangle, text height=1.5ex, text depth=0.25ex, yshift=0.5mm, fill=white, draw=black, minimum height=3mm, yshift=-0.5mm, minimum width=3mm, font={\small}, tikzit category=circuit, inner sep=2pt]
\tikzstyle{big gate}=[shape=rectangle, text height=1.5ex, text depth=0.25ex, yshift=0.5mm, fill=white, draw=black, minimum height=10mm, yshift=-0.5mm, minimum width=5mm, font={\small}, tikzit category=circuit]
\tikzstyle{Z dot}=[inner sep=0mm, minimum size=2mm, shape=circle, draw=black, fill={rgb,255: red,221; green,255; blue,221}, tikzit category=zx]
\tikzstyle{Z phase dot}=[minimum size=5mm, font={\footnotesize\boldmath}, shape=rectangle, rounded corners=2mm, inner sep=0.2mm, outer sep=-2mm, scale=0.8, tikzit shape=circle, draw=black, fill={rgb,255: red,221; green,255; blue,221}, tikzit draw=blue, tikzit category=zx]
\tikzstyle{X dot}=[Z dot, shape=circle, draw=black, fill={rgb,255: red,255; green,136; blue,136}, tikzit category=zx]
\tikzstyle{X phase dot}=[Z phase dot, tikzit shape=circle, tikzit draw=blue, fill={rgb,255: red,255; green,136; blue,136}, font={\footnotesize\boldmath}, tikzit category=zx]
\tikzstyle{hadamard}=[fill=yellow, draw=black, shape=rectangle, inner sep=0.6mm, minimum height=1.5mm, minimum width=1.5mm, tikzit category=zx]
\tikzstyle{paulibox}=[fill={rgb,255: red,221; green,221; blue,255}, draw=black, shape=rectangle, inner sep=0.6mm, minimum height=5mm, minimum width=5mm, font={\footnotesize}, text height=1.5ex, text depth=0.25ex, tikzit category=zx]
\tikzstyle{vertex}=[inner sep=0mm, minimum size=1mm, shape=circle, draw=black, fill=black, tikzit category=misc]
\tikzstyle{vertex set}=[inner sep=0mm, minimum size=1mm, shape=circle, draw=black, fill=white, font={\footnotesize\boldmath}, tikzit category=misc]
\tikzstyle{small black dot}=[fill=black, draw=black, shape=circle, inner sep=0pt, minimum width=1.2mm, tikzit category=circuit]
\tikzstyle{cnot ctrl}=[fill=black, draw=black, shape=circle, inner sep=0pt, minimum width=1.2mm, tikzit category=circuit]
\tikzstyle{cnot targ}=[fill=white, draw=white, shape=circle, tikzit category=circuit, label={center:$\oplus$}, inner sep=0pt, minimum width=2.1mm, tikzit fill={rgb,255: red,102; green,204; blue,255}, tikzit draw=black]
\tikzstyle{ket}=[fill=white, draw=black, shape=regular polygon, regular polygon sides=3, regular polygon rotate=-30, scale=0.7, inner sep=1pt, tikzit category=circuit, tikzit shape=rectangle, tikzit fill=green]
\tikzstyle{bra}=[fill=white, draw=black, shape=regular polygon, regular polygon sides=3, regular polygon rotate=30, scale=0.7, inner sep=1pt, tikzit category=circuit, tikzit shape=rectangle, tikzit fill=red]
\tikzstyle{scalar}=[shape=rectangle, text height=1.5ex, text depth=0.25ex, yshift=0.5mm, fill=white, draw=black, minimum height=5mm, yshift=-0.5mm, minimum width=5mm, font={\small}]
\tikzstyle{clabel}=[fill=white, draw=none, shape=rectangle, tikzit fill={rgb,255: red,56; green,255; blue,242}, font={\footnotesize}, inner sep=1pt, tikzit category=labels]
\tikzstyle{empty diagram}=[draw={gray!40!white}, dashed, shape=rectangle, minimum width=1cm, minimum height=1cm, tikzit category=misc]
\tikzstyle{amap}=[fill=white, draw=black, shape=NEbox, tikzit category=asymmetric, tikzit fill=yellow, tikzit shape=rectangle]
\tikzstyle{amap conj}=[fill=white, draw=black, shape=NWbox, tikzit category=asymmetric, tikzit fill=green, tikzit shape=rectangle]
\tikzstyle{amap adj}=[fill=white, draw=black, shape=SEbox, tikzit category=asymmetric, tikzit fill=red, tikzit shape=rectangle]
\tikzstyle{amap trans}=[fill=white, draw=black, shape=SWbox, tikzit category=asymmetric, tikzit fill=orange, tikzit shape=rectangle]
\tikzstyle{astate}=[fill=white, draw=black, shape=NEtriangle, tikzit category=asymmetric, tikzit shape=circle, tikzit fill=yellow]
\tikzstyle{astate conj}=[fill=white, draw=black, shape=NWtriangle, tikzit category=asymmetric, tikzit shape=circle, tikzit fill=green]
\tikzstyle{astate adj}=[fill=white, draw=black, shape=SEtriangle, tikzit category=asymmetric, tikzit shape=circle, tikzit fill=red]
\tikzstyle{astate trans}=[fill=white, draw=black, shape=SWtriangle, tikzit category=asymmetric, tikzit shape=circle, tikzit fill=orange]
\tikzstyle{white dot}=[inner sep=0mm, minimum size=2mm, shape=circle, draw=black, fill={rgb,255: red,250; green,250; blue,250}]
\tikzstyle{white phase dot}=[minimum size=5mm, font={\footnotesize\boldmath}, shape=rectangle, rounded corners=2mm, inner sep=0.2mm, outer sep=-2mm, scale=0.8, tikzit shape=circle, draw=black, fill={rgb,255: red,250; green,250; blue,250}, tikzit draw=blue]
\tikzstyle{hbox}=[shape=rectangle, text height=2mm, fill={rgb,255: red,255; green,235; blue,61}, draw=black, minimum height=2mm, minimum width=2mm, font={\small}, tikzit category=zh, inner sep=0pt, rounded corners=0.5mm]
\tikzstyle{Z dot (zh)}=[inner sep=0mm, minimum size=2mm, shape=circle, draw=black, fill={rgb,255: red,250; green,250; blue,250}, tikzit category=zh]
\tikzstyle{X dot (zh)}=[Z dot, shape=circle, draw=black, fill={rgb,255: red,193; green,193; blue,193}, tikzit category=zh]
\tikzstyle{triangle}=[fill={rgb,255: red,255; green,136; blue,136}, draw=black, shape=isosceles triangle, isosceles triangle apex angle=60, minimum size=2.5mm, inner sep=0mm]
\tikzstyle{labelled hbox}=[shape=rectangle, text height=1.75ex, text depth=0.5ex, fill={rgb,255: red,255; green,235; blue,61}, draw=black, minimum height=3mm, minimum width=4mm, font={\small}, tikzit category=zh, inner sep=1.3pt, rounded corners=0.5mm]
\tikzstyle{Z phase dot (zh)}=[Z phase dot, tikzit shape=circle, tikzit draw=blue, fill={rgb,255: red,250; green,250; blue,250}, font={\footnotesize\boldmath}, tikzit category=zh]
\tikzstyle{X phase dot (zh)}=[Z phase dot, tikzit shape=circle, tikzit draw=blue, fill={rgb,255: red,193; green,193; blue,193}, font={\footnotesize\boldmath}, tikzit category=zh]
\tikzstyle{W node}=[fill=black, draw=black, shape=regular polygon, regular polygon sides=3, minimum size=2mm]
\tikzstyle{Z dot (zw)}=[fill=white, draw=black, shape=circle, minimum width=1.2mm, inner sep=0pt]
\tikzstyle{Z phase dot XL}=[Z phase dot, fill={rgb,255: red,250; green,250; blue,250}, draw=black, shape=circle, tikzit draw={rgb,255: red,191; green,0; blue,64}, tikzit shape=circle, font={\large\boldmath}, inner sep=0.0mm]

\tikzstyle{hadamard edge}=[-, dashed, dash pattern=on 2pt off 0.5pt, thick, draw={rgb,255: red,68; green,136; blue,255}]
\tikzstyle{box edge}=[-, dashed, dash pattern=on 2pt off 0.5pt, thick, draw={rgb,255: red,203; green,192; blue,225}]
\tikzstyle{brace edge}=[-, tikzit draw=blue, decorate, decoration={brace,amplitude=1mm,raise=-1mm}]
\tikzstyle{diredge}=[->, thick]
\tikzstyle{double edge}=[-, double, shorten <=-1mm, shorten >=-1mm, double distance=2pt]
\tikzstyle{gray edge}=[-, {gray!60!white}]
\tikzstyle{pointer edge}=[->, very thick, gray]
\tikzstyle{boldedge}=[-, line width=1.0pt, shorten <=-0.17mm, shorten >=-0.17mm]
\tikzstyle{bidir edge}=[<->, very thick, draw={rgb,255: red,191; green,191; blue,191}]
\tikzstyle{purple edge}=[->, thick, draw={rgb,255: red,225; green,117; blue,216}]
\tikzstyle{green edge}=[->, thick, draw={rgb,255: red,167; green,231; blue,137}]
\tikzstyle{orange edge}=[->, thick, draw={rgb,255: red,245; green,170; blue,63}]
\tikzstyle{blue edge}=[->, thick, draw={rgb,255: red,68; green,136; blue,255}]
\tikzstyle{any edge}=[->, thick, draw=cyan]
\tikzstyle{red edge}=[->, thick, draw={rgb,255: red,255; green,136; blue,136}]
\tikzstyle{bidiredge}=[<->, thick]
\tikzstyle{dashed diredge}=[->, dashed, dash pattern=on 1pt off 0.5pt]
\tikzstyle{bidashed diredge}=[<->, dashed, dash pattern=on 1pt off 0.5pt]
\tikzstyle{gray fill}=[-, fill={rgb,255: red,234; green,234; blue,234}, draw=black]
\tikzstyle{white fill}=[-, fill=white]

\input{quantum.tikzdefs}
\usepackage{booktabs}

\newtheorem{theorem}{Theorem}
\newtheorem{task}{Task}
\newtheorem{lemma}{Lemma}

\newcommand{\trace}{\operatorname{Tr}}

\title{Any Clifford+T circuit can be controlled with constant T-depth overhead}
\author{Isaac H. Kim}
\affil{Department of Computer Science, University of California, Davis, CA, 95616, USA}
\author{Tuomas Laakkonen}
\affil{Plasma Science and Fusion Center, Massachusetts Institute of Technology, Cambridge, MA, 02139, USA}

\begin{document}

\maketitle

\begin{abstract}
    \noindent Since an $n$-qubit circuit consisting of CNOT gates can have up to $\Omega(n^2/\log n)$ CNOT gates, it is natural to expect that $\Omega(n^2 / \log n)$ Toffoli gates are needed to apply a controlled version of such a circuit. We show that the Toffoli count can be reduced to at most $n$. The Toffoli depth can also be reduced to $O(1)$, at the cost of $2n$ Toffoli gates, even without using any ancilla or measurement. In fact, using a measurement-based uncomputation, the Toffoli depth can be further reduced to $1$. From this, we give two corollaries: any controlled Clifford circuit can be implemented with $O(1)$ T-depth, and any Clifford+T circuit with T-depth $D$ can be controlled with T-depth $O(D)$, even without ancillas. As an application, we show how to catalyze a rotation by any angle up to precision $\epsilon$ in T-depth exactly $1$ using a universal $\lceil\log_2(8/\epsilon)\rceil$-qubit catalyst state.
\end{abstract}

\section{Introduction}
\label{sec:introduction}

\begin{table}[t]
    \centering
    \begin{tabular}{lllllll}
    \toprule
        Reference & Circuit Type & Non-Clifford Depth & Non-Clifford Count & \multicolumn{2}{c}{Optimal?} & Ancillas \\ 
    \midrule
        Section \ref{subsec:cnottoffolicount} & CNOT & $n - c$ & $n - c$ & No & $\approx$ & - \\
        Section \ref{subsec:cnotdepthconstant} & CNOT & $12$ & $2n$ & $\approx$ & $\approx$ & - \\
        Section \ref{sec:cnot_toffoli_one} & CNOT & $1$ & $n$ & Yes & $\approx$ & $2n - 1$ \\
        Section \ref{subsec:controlledclifford} & Clifford & $O(1)$ & $O(n)$ & $\approx$ & No & - \\
        Section \ref{subsec:controlledunitary} & Clifford+T & $O(D)$ & $O(C + n)$ & $\approx$ & No  & - \\
    \bottomrule
    \end{tabular}
    \caption{A summary of the results presented in this paper, indexed by the type of circuit they apply to. For CNOT circuits, the non-Clifford Depth/Count refers to Toffoli gates, while for Clifford and Clifford+T circuits, it refers to T-gates. The two columns beneath the `Optimal?' heading refer to whether the Non-Clifford Depth and Count are optimized, respectively, where the symbol `$\approx$' indicates approximate optimality (that is, always within a constant factor of optimal). We define $n$ as the number of qubits in each circuit, $D$ and $C$ as the T-depth and T-count of the uncontrolled circuit, and $c \geq 0$ is a factor depending on the specific controlled CNOT circuit being synthesized.}
    \label{tab:results}
\end{table}

Fault-tolerant computation schemes based on quantum error-correcting codes will be necessary to achieve large-scale quantum computation. In many such schemes, the available set of quantum gates is not fully general but consists of the Clifford gates together with a non-Clifford gate, such as the $T$-gate or Toffoli gate. Due to the Eastin-Knill theorem~\cite{Eastin2009}, it is not possible to implement both Clifford and non-Clifford gates within the same error correcting code. A standard workaround for this is to implement Clifford operations within the code and rely on state teleportation for the non-Clifford gates. In this method, magic states are prepared outside of the code and applied via a measurement and classically controlled correction operation. Historically, non-Clifford gates have been the most expensive (in terms of time and space) to implement \cite{Litinski2019magic}, and so minimizing them is an important goal. This leads us to the first task we consider in this paper:

\begin{task}
    Given a description of a quantum circuit, produce an equivalent circuit with a reduced non-Clifford gate count.
    \label{task:count}
\end{task}

Furthermore, while in many topological codes (such as the surface code) Clifford computations and magic state preparation can be arbitrarily parallelized, given a large enough supply of qubits \cite{fowler2013time,litinski2019game}, this is not always the case for classically-controlled operations used to inject magic states, which must be processed serially if a measurement basis must be updated based on previous measurement outcomes. Therefore, minimizing the total time required in the regime of many available qubits, and hence maximal parallelization, is essentially the same as minimizing the depth of non-Clifford gates. While minimizing the non-Clifford depth had not been the top priority in the literature due to the large space-time cost of such gates~\cite{litinski2019game,Litinski2019magic}, there has been recent progress in reducing the spacetime cost of logical magic states~\cite{Itogawa_2025,gidney2024magic,Daguerre:2024gjd,sahay2025foldtransversalsurfacecodecultivation,claes2025cultivatingtstatessurface,Vaknin:2025pbp,Chen:2025imz,Daguerre:2025boq,Rosenfeld:2025xvf}. Thus there may be practical motivations to minimize non-Clifford depth; see also Ref.~\cite{mcardle2025fastcuriousacceleratefaulttolerant}. This leads us to consider a second task:

\begin{task}
    Given a description of a quantum circuit, produce an equivalent circuit with a reduced non-Clifford gate depth.
    \label{task:depth}
\end{task}

We tackle these two tasks for the special case of controlled circuits, which are defined by:
$$ c(U) [\ket{c} \otimes \ket{\psi}] = \ket{c} \otimes U^c \ket{\psi} $$
Controlled circuits are ubiquitous in quantum algorithm -- for example, in the Hadamard test, quantum phase estimation, block-encoding constructions, and other applications discussed in Section \ref{sec:applications} -- but a naive gate-by-gate compilation yields very poor scaling in terms of non-Clifford gate count and depth. This overhead is most pronounced when considering Clifford circuits, as their controlled versions must be non-Clifford. In this work, we find constructions for controlled Clifford circuits that are asymptotically better in both non-Clifford count and depth than naive compilation. That this is possible is not necessarily surprising, but explicit constructions were not previously known. Indeed, because of their unique algebraic structure, Clifford circuits have a history of admitting good solutions to various circuit synthesis tasks (for instance \cite{maslov2022depth,murphy2023global,meijer2023dynamic,goubault2022decoding,kutin2007computation}), in contrast with more general quantum circuits, for which many of these tasks are difficult in a complexity-theoretic sense \cite{vandewetering2024optimising}.

In Section \ref{sec:cnots}, we start by considering circuits composed only of CNOT gates, which are a subset of Clifford circuits. Since an arbitrary CNOT circuit may contain up to $\Omega(n^2 / \log n)$ gates \cite{patel2003efficient}, it is natural to expect that $\Omega(n^2 / \log n)$ non-Clifford gates are needed to implement a controlled version. However, we show that this is not the case, and that at most $n - c$ Toffoli gates are needed, where $c \geq 0$ is a factor that depends on the specific circuit; in fact, we prove that this is at most 50\% more Toffoli gates than the optimal circuit. 

Furthermore, the Toffoli-depth can be reduced all the way to $O(1)$ while preserving the approximate optimality in terms of Toffoli-count; without ancilla qubits, we can reduce the Toffoli-depth to a constant, and if we are allowed $2n - 1$ clean ancillas, we can reduce the $T$-depth to exactly $1$, which is strictly optimal. This requires the use of measurement-based uncomputation \cite{Gidney2018halvingcostof}, but we ensure that the classically-controlled correction operations can be parallelized. 

Then in Section \ref{sec:gen}, based on the normal form for Clifford circuits of \cite{aaronson2004improved}, we extend this to general Clifford circuits, to show that controlled Clifford circuits may be implemented with constant $T$-depth, and finally to all Clifford+T circuits. In this case, we show that, given an $n$-qubit circuit with $T$-depth $D$ and $T$-count $C$, we can produce a controlled version of it with $T$-depth $O(D)$ and $T$-count $O(C + n)$. A summary of these results is shown in Table \ref{tab:results}.

Finally, in Section \ref{sec:applications} we consider concrete applications of our constructions. In particular, multivariate trace estimation problems \cite{Quek2024} and, more significantly, arbitrary-angle rotation synthesis by catalysis. Since arbitrary-angle rotation gates are not usually available in fault-tolerant gate sets, they must be approximated from Clifford+T operations. While it is known how to approximate any particular rotation gate to precision $\epsilon$ essentially optimally with $O(\log \frac{1}{\epsilon})$ gates \cite{ross2016optimal}, applying this to approximate a circuit of $k$ rotation gates to precision $\epsilon$ requires a $T$-depth of $O(k\log\frac{k}{\epsilon})$ (because a precision of $\frac{\epsilon}{k}$ in each gate is required). Thus, this method has a more-than-constant overhead in $T$-depth that we would hope to avoid. 

By contrast, in \cite{kim2025catalytic} it was shown that by using a catalyst state (which is an auxiliary state that must be prepared but is not consumed), the $T$-depth can be reduced to $3k + O(\mathrm{polylog}(\frac{k}{\epsilon}))$ with $O(\log^2\frac{k}{\epsilon})$ ancilla qubits. Using our controlled CNOT construction, we improve this to $k + \tilde{O}(\log\frac{k}{\epsilon})$ with $O(\log\frac{k}{\epsilon})$ ancilla qubits. This gives an essentially optimal method of compiling Clifford+$R_Z$ circuits to the Clifford+T gate set, in terms of $T$-depth up to an additive logarithmic factor, and resolves an open question of \cite{kim2025catalytic}. A comparison with alternative methods is given in Table \ref{tab:catalysis}.

\begin{table}[t]
    \centering
    \begin{tabular}{llll}
    \toprule
        Algorithm & $T$-depth & $T$-count & Qubits \\
    \midrule
        Solovay-Kitaev \cite{dawson2006solovay} & $O(k\log^{3.97}\frac{k}{\epsilon})$ & $O(k\log^{3.97}\frac{k}{\epsilon})$ &  $n$  \\
        Ross-Selinger \cite{ross2016optimal} & $3k\log_2\frac{k}{\epsilon} + \tilde{O}(1)$ &  $3k\log_2\frac{k}{\epsilon} + \tilde{O}(1)$  &  $n$  \\
        Gidney \cite{Gidney2018halvingcostof} & $O(k\log\log\frac{k}{\epsilon} + \log^2\frac{k}{\epsilon})$ & $O(k\log\frac{k}{\epsilon} + \log^2\frac{k}{\epsilon})$ &  $n + O(\log\frac{k}{\epsilon})$ \\
        Kim \cite{kim2025catalytic} & $3k + O(\mathrm{polylog}\frac{k}{\epsilon})$ & $O(k\log^2\frac{k}{\epsilon} + \mathrm{polylog}\frac{k}{\epsilon})$ & $n + O(\log^2\frac{k}{\epsilon})$\\
        \emph{This Paper} &  $k + \tilde{O}(\log\frac{k}{\epsilon})$ & $O(k\log\frac{k}{\epsilon}) + \tilde{O}(\log\frac{k}{\epsilon})$ & $n + O(\log\frac{k}{\epsilon})$\\
    \bottomrule
    \end{tabular}
    \caption{A comparison of arbitrary-angle $Z$-rotation synthesis schemes. Specifically, given an $n$-qubit Clifford+$R_Z$ circuit with $k$ $Z$-rotation gates, the $T$-depths, $T$-counts, and number of qubits required to transpile to the Clifford+T gateset are listed for several algorithms. Our method (discussed in Section \ref{sec:catalysis}) is the most competitive for $T$-depth, and simultaneously asymptotically optimal in $T$-count, but is outperformed in qubit count by ancilla-free methods based on individual gate approximations. Note that $\tilde{O}(.)$ hides polynomial factors in $\log\log\frac{1}{\epsilon}$.} 
    \label{tab:catalysis}
\end{table}

\section{Preliminaries}
\label{sec:preliminaries}

Here, we set up our notation and review pertinent facts about Clifford and non-Clifford gates. Cliffords are unitaries that map Paulis to Paulis under conjugation. It is a well-known fact that such unitaries form a group and that the group is generated by $H, S,$ and CNOT. Throughout this paper, we will consider Clifford circuits acting on $n$ qubits.

Various standard forms for such Clifford circuits have been proposed, for instance \cite{aaronson2004improved, duncan2020graph, maslov2018shorter}. Although these forms differ in details, they share a common theme: that an arbitrary $n$-qubit Clifford is decomposed into $O(1)$ layers, each of which consists of gates of the same type. The gates are chosen from the standard set $\{H, S, \text{CNOT}, \text{CZ} \}$. For the purpose of achieving Tasks~\ref{task:count} and \ref{task:depth}, the layers consisting of single-qubit gates pose little difficulty. However, less obvious is how to optimize the layers consisting of two-qubit gates, such as CNOTs. As such, we will primarily focus on CNOT circuits.

\subsection{CNOT Circuits as Binary Matrices}
\label{subsec:cnot_layer_binary_matrix}

CNOT circuits can be viewed as invertible binary matrices acting on a length-$n$ binary vectors. Let $\vec{x}\in \mathbb{F}_2^n$ and let $U$ be a unitary that represents a CNOT circuit. There is an $n\times n$ matrix $A$ such that:
\begin{equation}
    U|\vec{x}\rangle = |A\vec{x}\rangle \label{eq:cnot_circuit_binary_matrix}
\end{equation}
We call this matrix as the \emph{parity matrix}. If we apply a CNOT gate on qubits $i$ and $j$, to circuit $C$ on the left, so that $C' = \mathrm{CNOT}_{ij} \cdot C$, the parity matrix changes by adding row $i$ to $j$:
$$
A_{C'} = \begin{pmatrix}
    \ddots \\
    & 1 & \cdots & 0  \\
    & \vdots & & \vdots \\
    & 1 & \cdots & 1 \\
    & & & & \ddots
\end{pmatrix} A_C = \begin{pmatrix}
    & \vdots & \\
    (A_C)_{i1} & \cdots & (A_C)_{in} \\
    & \vdots & \\
    (A_C)_{j1} \oplus (A_C)_{i1} & \cdots & (A_C)_{jn} \oplus (A_C)_{in} \\
    &\vdots & 
\end{pmatrix}
$$
Likewise, applying a CNOT gate from the right corresponds to adding column $i$ to column $j$. Moreover, appending two circuits corresponds to multiplying both their unitaries and their parity matrices:
\ctikzfig{parity-matrix-multiply}
Composing two CNOT circuits in parallel corresponds to the direct sum of their parity matrices and the tensor product of their unitaries:
\ctikzfig{parity-matrix-direct-sum}
Thus, all parity matrices are invertible. Given any CNOT circuit, deriving the corresponding parity matrix can be done easily by starting with the identity matrix (corresponding to the empty circuit) and applying these row operations incrementally. In the other direction, a CNOT circuit can be synthesized from any invertible $n \times n$ matrix over $\mathbb{F}_2$ using Gaussian elimination. 

\subsection{Controlled Unitaries}

For any unitary $U$ on $n$ qubits, we define the controlled unitary $c(U)$ on $n+1$ qubits by:
$$
c(U)(\ket{0} \otimes I_n) = \ket{0} \otimes I_n \qquad c(U)(\ket{1} \otimes I_n) = \ket{1} \otimes U
$$
It is denoted in circuit notation as
\ctikzfig{controlled-circuit-notation}
and we will use the following identities:
$$
c(UV) = c(U)c(V) \qquad c(UVU^{-1}) = (I \otimes U)c(V)(I \otimes U^{-1})
$$

\section{Controlled CNOT Circuits}
\label{sec:cnots}

\subsection{Approximately optimal Toffoli-count}
\label{subsec:cnottoffolicount}

Our construction for CNOT circuits is based on a particular matrix decomposition called the \emph{generalized Jordan normal form}. This is a decomposition that can be computed for square matrices over any field, using only arithmetic within that field, which directly generalizes the Jordan normal form for complex matrices. It is a normal form in the sense that two matrices $A$ and $B$ have the same generalized Jordan normal form if and only if they are similar -- that is, $A = SBS^{-1}$ for some invertible $S$ -- so that $A$ and $B$ represent the same action in different bases. In particular, we have the following.

\begin{theorem}[Generalized Jordan Normal Form]
    \label{thm:rationalcanonicalform}
    Given an $n \times n$ matrix $A$ over an arbitrary field $\mathbb{F}$, there exists $k$ polynomials $f_1, f_2, \dots, f_k \in \mathbb{F}[x]$, and an invertible matrix $S \in GL_n(\mathbb{F})$ such that
    $$ SAS^{-1} = \begin{pmatrix}
        C_{f_1} \\
        & C_{f_2} \\
        & & \ddots \\
        & & & C_{f_k}
    \end{pmatrix} \qquad C_{f_i} = \begin{pmatrix}
        & & & & -a_{i0} \\
        1 & & & & -a_{i1} \\
        & 1 & & & -a_{i2} \\
        & & \ddots & & \vdots \\
        & & & 1 & -a_{i,d_i-1}
    \end{pmatrix}$$
    where $C_{f_i}$ is the companion matrix corresponding to $f_i = a_{i0} + a_{i1}x + a_{i2}x^2 + \cdots + a_{i,d_i-1}x^{d_i-1} + x^{d_i}$. Moreover, we have $f_i(x) = q_i(x)^k$ where each $q_i(x)$ is irreducible and each $f_i$ is called an \emph{elementary divisor} of $A$.
\end{theorem}

\begin{proof}
    There are many proofs in the literature, for example \cite[Chapter 10, Theorem 24]{birkhoff2017survey}.
\end{proof}

Since each $C_{f_i}$ is sparse, this immediately shows the remarkable fact that \emph{every square matrix is similar to a sparse matrix}. This applies not just for diagonalizable matrices over $\mathbb{C}$, where this is given by the eigendecomposition, but for arbitrary matrices over arbitrary fields. We will exploit this by applying this decomposition to the parity matrix of a CNOT circuit as a matrix over $\mathbb{F}_2$. We will show that the resulting sparse matrix can be synthesized with few CNOTs, and hence synthesize controlled CNOT circuits with few Toffoli gates. Firstly, we consider each $C_{f_i}$ individually, with a special case for $f_i(x) = (1 + x)^{d_i}$.

\begin{lemma}
    \label{thm:controlledcompanionmatrix}
    For any polynomial $f(x) \in \mathbb{F}_2[x]$ of degree $d \geq 2$ with $f(0) = 1$, the controlled CNOT circuit with parity matrix $C_f$ can be implemented with $d$ Toffoli gates.
\end{lemma}

\begin{proof}
    First note that if $f(0) = 1$, we must have $a_{i0} = 1$. If $a_{ij} = 0$ for $j > 0$, the circuit with this parity matrix is given by a cyclic shift by one qubit, which can be represented either with SWAP gates, or with triples of CNOT gates:
    $$
        
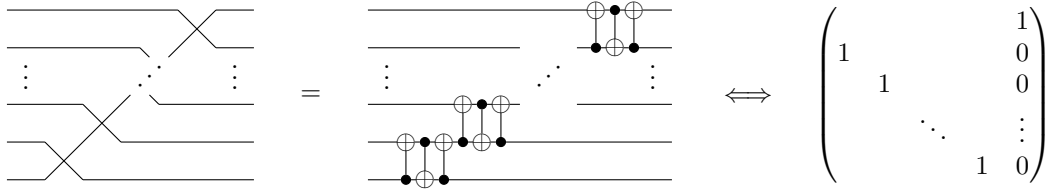
\begin{tikzpicture}
	\begin{pgfonlayer}{nodelayer}
		\node [style=none] (0) at (1, -1.25) {};
		\node [style=none] (1) at (1, -2.25) {};
		\node [style=none] (2) at (2, -1.25) {};
		\node [style=none] (3) at (2, -2.25) {};
		\node [style=none] (6) at (2, -0.25) {};
		\node [style=none] (7) at (2, -1.25) {};
		\node [style=none] (8) at (3.25, 0) {};
		\node [style=none] (9) at (3, -1.25) {};
		\node [style=none] (12) at (4.5, 2.25) {};
		\node [style=none] (13) at (4.25, 1) {};
		\node [style=none] (14) at (5.5, 2.25) {};
		\node [style=none] (15) at (5.5, 1.25) {};
		\node [style=none] (18) at (3.55, 0.525) {\rotatebox{82}{$\ddots$}};
		\node [style=none] (19) at (0, 2.25) {};
		\node [style=none] (20) at (0, 1.25) {};
		\node [style=none] (21) at (0, -0.25) {};
		\node [style=none] (22) at (6.5, -2.25) {};
		\node [style=none] (23) at (6.5, -1.25) {};
		\node [style=none] (24) at (6.5, -0.25) {};
		\node [style=none] (25) at (0, -2.25) {};
		\node [style=none] (26) at (0, -1.25) {};
		\node [style=none] (27) at (6.5, 1.25) {};
		\node [style=none] (28) at (6.5, 2.25) {};
		\node [style=none] (29) at (6, 0.675) {$\vdots$};
		\node [style=none] (30) at (0.5, 0.675) {$\vdots$};
		\node [style=none] (31) at (3.5, 1.25) {};
		\node [style=none] (32) at (3.75, 1) {};
		\node [style=none] (33) at (4, -0.25) {};
		\node [style=none] (34) at (3.75, 0) {};
		\node [style=none] (35) at (8, 0) {$=$};
		\node [style=cnot targ] (36) at (10.5, -1.25) {};
		\node [style=cnot ctrl] (37) at (10.5, -2.25) {};
		\node [style=cnot ctrl] (38) at (12, -1.25) {};
		\node [style=cnot ctrl] (39) at (11.5, -2.25) {};
		\node [style=cnot targ] (40) at (12, -0.25) {};
		\node [style=cnot targ] (41) at (11.5, -1.25) {};
		\node [style=none] (42) at (13.5, -0.25) {};
		\node [style=cnot ctrl] (43) at (13, -1.25) {};
		\node [style=cnot targ] (44) at (15.5, 2.25) {};
		\node [style=none] (45) at (15, 1.25) {};
		\node [style=cnot targ] (46) at (16.5, 2.25) {};
		\node [style=cnot ctrl] (47) at (16.5, 1.25) {};
		\node [style=none] (48) at (14.05, 0.525) {\rotatebox{82}{$\ddots$}};
		\node [style=none] (49) at (9.5, 2.25) {};
		\node [style=none] (50) at (9.5, 1.25) {};
		\node [style=none] (51) at (9.5, -0.25) {};
		\node [style=none] (52) at (17.5, -2.25) {};
		\node [style=none] (53) at (17.5, -1.25) {};
		\node [style=none] (54) at (17.5, -0.25) {};
		\node [style=none] (55) at (9.5, -2.25) {};
		\node [style=none] (56) at (9.5, -1.25) {};
		\node [style=none] (57) at (17.5, 1.25) {};
		\node [style=none] (58) at (17.5, 2.25) {};
		\node [style=none] (59) at (17, 0.675) {$\vdots$};
		\node [style=none] (60) at (10, 0.675) {$\vdots$};
		\node [style=none] (61) at (13.5, 1.25) {};
		\node [style=none] (63) at (15, -0.25) {};
		\node [style=cnot targ] (65) at (13, -0.25) {};
		\node [style=cnot ctrl] (66) at (15.5, 1.25) {};
		\node [style=cnot targ] (67) at (11, -2.25) {};
		\node [style=cnot targ] (68) at (12.5, -1.25) {};
		\node [style=cnot ctrl] (69) at (12.5, -0.25) {};
		\node [style=cnot ctrl] (70) at (11, -1.25) {};
		\node [style=cnot ctrl] (71) at (16, 2.25) {};
		\node [style=cnot targ] (72) at (16, 1.25) {};
	\end{pgfonlayer}
	\begin{pgfonlayer}{edgelayer}
		\draw (0.center) to (3.center);
		\draw (6.center) to (9.center);
		\draw (12.center) to (15.center);
		\draw (12.center) to (19.center);
		\draw (21.center) to (6.center);
		\draw (23.center) to (9.center);
		\draw (3.center) to (22.center);
		\draw (28.center) to (14.center);
		\draw (15.center) to (27.center);
		\draw (26.center) to (0.center);
		\draw (1.center) to (25.center);
		\draw (31.center) to (20.center);
		\draw (32.center) to (31.center);
		\draw (33.center) to (34.center);
		\draw (33.center) to (24.center);
		\draw (8.center) to (7.center);
		\draw (7.center) to (1.center);
		\draw (14.center) to (13.center);
		\draw (44) to (49.center);
		\draw (51.center) to (40);
		\draw (53.center) to (43);
		\draw (39) to (52.center);
		\draw (58.center) to (46);
		\draw (47) to (57.center);
		\draw (56.center) to (36);
		\draw (37) to (55.center);
		\draw (61.center) to (50.center);
		\draw (63.center) to (54.center);
		\draw (65) to (42.center);
		\draw (66) to (45.center);
		\draw (37) to (36);
		\draw (41) to (39);
		\draw (40) to (38);
		\draw (38) to (41);
		\draw (65) to (43);
		\draw (70) to (67);
		\draw (69) to (68);
		\draw (36) to (70);
		\draw (41) to (70);
		\draw (39) to (67);
		\draw (67) to (37);
		\draw (38) to (68);
		\draw (68) to (43);
		\draw (65) to (69);
		\draw (69) to (40);
		\draw (72) to (71);
		\draw (71) to (46);
		\draw (71) to (44);
		\draw (66) to (72);
		\draw (72) to (47);
		\draw (46) to (47);
		\draw (66) to (44);
	\end{pgfonlayer}
\end{tikzpicture} \quad\iff\quad \begin{pmatrix}
              &   &        &   & 1 \\
            1 &   &        &   & 0 \\
              & 1 &        &   & 0 \\
              &   & \ddots &   & \vdots \\
              &   &        & 1 & 0
        \end{pmatrix}
    $$
    This can be made controlled by replacing the middle CNOT of each SWAP gate with a Toffoli gate:
    \ctikzfig{cyclic-permutation-controlled}
    Suppose that we have $a_{ij} = 1$ for some $j > 0$, then we can perform row operations using CNOT gates to copy the one in the top-right corner of the above parity matrix to wherever it is required. Here we write $a_{ij}$ next to each CNOT to indicate that it should be generated only when $a_{ij}  = 1$.
    $$
        \begin{pmatrix}
                &   &        &   & 1 \\
              1 &   &        &   & a_{i1} \\
                & 1 &        &   & a_{i2} \\
                &   & \ddots &   & \vdots \\
                &   &        & 1 & a_{i,d-1}
        \end{pmatrix} \quad\iff\quad 
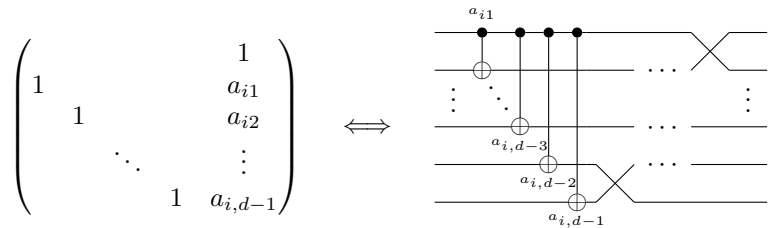
\begin{tikzpicture}
	\begin{pgfonlayer}{nodelayer}
		\node [style=none] (0) at (-4.75, 0.925) {$\vdots$};
		\node [style=cnot ctrl] (1) at (-4, 2.5) {};
		\node [style=cnot targ] (2) at (-4, 1.5) {};
		\node [style=cnot targ] (3) at (-3, 0) {};
		\node [style=cnot targ] (4) at (-2.25, -1) {};
		\node [style=cnot targ] (5) at (-1.5, -2) {};
		\node [style=cnot ctrl] (6) at (-3, 2.5) {};
		\node [style=cnot ctrl] (7) at (-2.25, 2.5) {};
		\node [style=cnot ctrl] (8) at (-1.5, 2.5) {};
		\node [style=none] (9) at (-5.25, 2.5) {};
		\node [style=none] (10) at (-5.25, 1.5) {};
		\node [style=none] (11) at (-5.25, 0) {};
		\node [style=none] (12) at (-5.25, -2) {};
		\node [style=none] (13) at (-5.25, -1) {};
		\node [style=none] (14) at (-3.5, 0.95) {\rotatebox{-17}{$\ddots$}};
		\node [style=none] (15) at (-4, 3) {\tiny $a_{i1}$};
		\node [style=none] (16) at (-3, -0.5) {\tiny $a_{i,d-3}$};
		\node [style=none] (17) at (-2.25, -1.5) {\tiny $a_{i,d-2}$};
		\node [style=none] (18) at (-1.5, -2.5) {\tiny $a_{i,d-1}$};
		\node [style=none] (19) at (-1, -1) {};
		\node [style=none] (20) at (-1, -2) {};
		\node [style=none] (21) at (0, -1) {};
		\node [style=none] (22) at (0, -2) {};
		\node [style=none] (23) at (0, 0) {};
		\node [style=none] (24) at (0, -1) {};
		\node [style=none] (25) at (1.5, -1) {};
		\node [style=none] (26) at (1.5, 2.5) {};
		\node [style=none] (27) at (1.5, 1.5) {};
		\node [style=none] (28) at (2.5, 2.5) {};
		\node [style=none] (29) at (2.5, 1.5) {};
		\node [style=none] (30) at (-1.25, 2.5) {};
		\node [style=none] (31) at (-1.25, 1.5) {};
		\node [style=none] (32) at (-1.25, 0) {};
		\node [style=none] (33) at (3.5, -2) {};
		\node [style=none] (34) at (3.5, -1) {};
		\node [style=none] (35) at (3.5, 0) {};
		\node [style=none] (36) at (-1.25, -2) {};
		\node [style=none] (37) at (-1.25, -1) {};
		\node [style=none] (38) at (3.5, 1.5) {};
		\node [style=none] (39) at (3.5, 2.5) {};
		\node [style=none] (40) at (3, 0.925) {$\vdots$};
		\node [style=none] (41) at (0, 1.5) {};
		\node [style=none] (42) at (1.5, 0) {};
		\node [style=none] (43) at (0.8, 1.475) {$\cdots$};
		\node [style=none] (44) at (0.8, -0.025) {$\cdots$};
		\node [style=none] (45) at (0.8, -1.025) {$\cdots$};
	\end{pgfonlayer}
	\begin{pgfonlayer}{edgelayer}
		\draw (8) to (5);
		\draw (4) to (7);
		\draw [in=90, out=-90] (6) to (3);
		\draw (1) to (2);
		\draw (1) to (6);
		\draw (6) to (7);
		\draw (7) to (8);
		\draw (10.center) to (2);
		\draw (1) to (9.center);
		\draw (11.center) to (3);
		\draw (4) to (13.center);
		\draw (12.center) to (5);
		\draw (19.center) to (22.center);
		\draw (26.center) to (29.center);
		\draw (26.center) to (30.center);
		\draw (32.center) to (23.center);
		\draw (34.center) to (25.center);
		\draw (22.center) to (33.center);
		\draw (39.center) to (28.center);
		\draw (29.center) to (38.center);
		\draw (37.center) to (19.center);
		\draw (20.center) to (36.center);
		\draw (41.center) to (31.center);
		\draw (42.center) to (35.center);
		\draw (24.center) to (20.center);
		\draw (28.center) to (27.center);
		\draw (36.center) to (5);
		\draw (37.center) to (4);
		\draw (32.center) to (3);
		\draw (31.center) to (2);
		\draw (30.center) to (8);
	\end{pgfonlayer}
\end{tikzpicture}
    $$
    This can also be written as the circuit below on the right, where the middle CNOT targets qubit $k$ such that $k$ is the smallest index where $a_{ik} = 1$ (in this diagram, it is drawn as if $k = 1$, but this is without loss of generality).
    \ctikzfig{companion-matrix-def}
    This follows from repeated application of the following identity:
    \ctikzfig{cnot-cancel-identity}
    Thus, all of these CNOTs can be controlled simply by controlling the middle CNOT.
    \ctikzfig{companion-matrix-controlled}
    Putting this together, we get a circuit for $c(C_{f})$ using only $d$ Toffoli gates.
\end{proof}

\begin{lemma}
    \label{thm:controlledxp1companion}
    For any $d$, the controlled CNOT circuit with parity matrix $C_{(x+1)^d}$ can be implemented with $d - 1$ Toffoli gates.
\end{lemma}

\begin{proof}
    Consider the $d\times d$ matrix $B_d$ defined below, which can be represented as CNOT circuit using $d - 1$ gates:
    $$ B_d = \begin{pmatrix}
        1 & \\
        1 & 1 \\
        & \ddots & \ddots \\
        & & 1 & 1
    \end{pmatrix} \Longleftrightarrow ~~
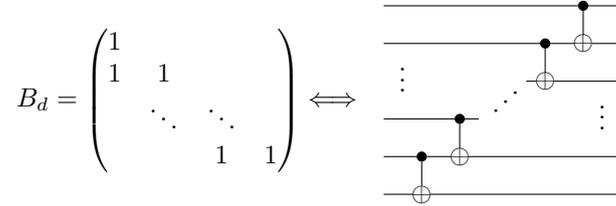
\begin{tikzpicture}
	\begin{pgfonlayer}{nodelayer}
		\node [style=none] (0) at (0, -1.5) {};
		\node [style=none] (1) at (0, -2.5) {};
		\node [style=cnot targ] (2) at (1, -2.5) {};
		\node [style=cnot ctrl] (3) at (1, -1.5) {};
		\node [style=cnot targ] (4) at (2, -1.5) {};
		\node [style=cnot ctrl] (5) at (2, -0.5) {};
		\node [style=cnot targ] (6) at (4.25, 0.5) {};
		\node [style=cnot ctrl] (7) at (4.25, 1.5) {};
		\node [style=none] (8) at (0, 1.5) {};
		\node [style=none] (9) at (0.5, 0.75) {$\vdots$};
		\node [style=none] (10) at (0, -0.5) {};
		\node [style=none] (11) at (2.5, -0.5) {};
		\node [style=none] (12) at (6.25, -1.5) {};
		\node [style=none] (13) at (6.25, -2.5) {};
		\node [style=none] (14) at (6.25, 1.5) {};
		\node [style=none] (15) at (5.75, -0.25) {$\vdots$};
		\node [style=none] (16) at (3.75, 0.5) {};
		\node [style=none] (17) at (6.25, 0.5) {};
		\node [style=none] (18) at (3, 0) {\rotatebox{80}{$\ddots$}};
		\node [style=cnot targ] (19) at (5.25, 1.5) {};
		\node [style=cnot ctrl] (20) at (5.25, 2.5) {};
		\node [style=none] (21) at (6.25, 2.5) {};
		\node [style=none] (22) at (0, 2.5) {};
	\end{pgfonlayer}
	\begin{pgfonlayer}{edgelayer}
		\draw (0.center) to (3);
		\draw (3) to (2);
		\draw (2) to (1.center);
		\draw (3) to (4);
		\draw (4) to (5);
		\draw (5) to (10.center);
		\draw (11.center) to (5);
		\draw (12.center) to (4);
		\draw (2) to (13.center);
		\draw (16.center) to (6);
		\draw (17.center) to (6);
		\draw (7) to (8.center);
		\draw (7) to (6);
		\draw (22.center) to (20);
		\draw (20) to (19);
		\draw (19) to (7);
		\draw (19) to (14.center);
		\draw (21.center) to (20);
	\end{pgfonlayer}
\end{tikzpicture}$$
    Then the characteristic polynomial of $B_d$ is equal to $(x + 1)^d$ (this can be seen by induction using Laplace expansion along the first row), and the minimal polynomial is equal to the characteristic polynomial (because the basis vector $\vec{e}_1$ is cyclic). Therefore, the generalized Jordan normal form of $B_d$ is exactly $C_{(x+1)^d}$, so there exists a matrix $M_d$ such that $C_{(x+1)^d} = M_d^{-1}B_dM_d$. Concretely, $M_d = M_d^{-1}$ is represented by a triangular pattern of CNOT gates, illustrated here for $M_5$:
    $$M_5 = \begin{pmatrix}  1 & 1 & 1 & 1 & 1 \\ 0 & 1 & 0 & 1 & 0 \\ 0 & 0 & 1 & 1 & 0 \\ 0 & 0 & 0 & 1 & 0 \\ 0 & 0 & 0 & 0 & 1\end{pmatrix} \Longleftrightarrow ~~
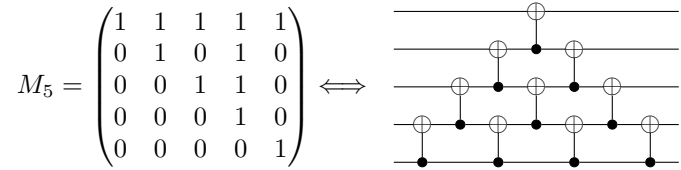
\begin{tikzpicture}
	\begin{pgfonlayer}{nodelayer}
		\node [style=cnot ctrl] (0) at (0.75, -2) {};
		\node [style=cnot targ] (1) at (0.75, -1) {};
		\node [style=cnot ctrl] (2) at (1.75, -1) {};
		\node [style=cnot targ] (3) at (1.75, 0) {};
		\node [style=cnot ctrl] (4) at (2.75, 0) {};
		\node [style=cnot targ] (5) at (2.75, 1) {};
		\node [style=cnot ctrl] (10) at (2.75, -2) {};
		\node [style=cnot targ] (11) at (2.75, -1) {};
		\node [style=cnot ctrl] (12) at (3.75, 1) {};
		\node [style=cnot targ] (13) at (3.75, 2) {};
		\node [style=cnot ctrl] (14) at (3.75, -1) {};
		\node [style=cnot targ] (15) at (3.75, 0) {};
		\node [style=cnot ctrl] (16) at (4.75, -2) {};
		\node [style=cnot targ] (17) at (4.75, -1) {};
		\node [style=cnot ctrl] (18) at (4.75, 0) {};
		\node [style=cnot targ] (19) at (4.75, 1) {};
		\node [style=cnot ctrl] (20) at (5.75, -1) {};
		\node [style=cnot targ] (21) at (5.75, 0) {};
		\node [style=cnot ctrl] (22) at (6.75, -2) {};
		\node [style=cnot targ] (23) at (6.75, -1) {};
		\node [style=none] (24) at (0, -2) {};
		\node [style=none] (25) at (0, -1) {};
		\node [style=none] (26) at (0, 0) {};
		\node [style=none] (27) at (0, 1) {};
		\node [style=none] (28) at (0, 2) {};
		\node [style=none] (29) at (7.5, -2) {};
		\node [style=none] (30) at (7.5, -1) {};
		\node [style=none] (31) at (7.5, 0) {};
		\node [style=none] (32) at (7.5, 1) {};
		\node [style=none] (33) at (7.5, 2) {};
	\end{pgfonlayer}
	\begin{pgfonlayer}{edgelayer}
		\draw (1) to (0);
		\draw (3) to (2);
		\draw (5) to (4);
		\draw (11) to (10);
		\draw (13) to (12);
		\draw (15) to (14);
		\draw (17) to (16);
		\draw (19) to (18);
		\draw (21) to (20);
		\draw (23) to (22);
		\draw (28.center) to (13);
		\draw (29.center) to (22);
		\draw (22) to (16);
		\draw (16) to (10);
		\draw (10) to (0);
		\draw (0) to (24.center);
		\draw (25.center) to (1);
		\draw (1) to (2);
		\draw (2) to (11);
		\draw (11) to (14);
		\draw (14) to (17);
		\draw (17) to (20);
		\draw (20) to (23);
		\draw (23) to (30.center);
		\draw (31.center) to (21);
		\draw (21) to (18);
		\draw (18) to (15);
		\draw (15) to (4);
		\draw (4) to (3);
		\draw (3) to (26.center);
		\draw (27.center) to (5);
		\draw (5) to (12);
		\draw (12) to (19);
		\draw (19) to (32.center);
		\draw (33.center) to (13);
	\end{pgfonlayer}
\end{tikzpicture}$$
    Let $U_C$, $U_M$ and $U_B$ be the CNOT circuits with parity matrices $C_{(x+1)^d}$, $M_d$, and $B_d$ respectively. Then since $B_d$ can be represented as a circuit of $d - 1$ CNOT gates, $c(U_C) = U_M^{-1}c(U_B)U_M$ can be implemented using $d - 1$ Toffoli gates.
\end{proof}

From this, we can construct controlled CNOT circuits in general, in a way that is \emph{approximately optimal} up to a factor of $\frac{3}{2}$, i.e, this construction uses at most $50$\% more Toffoli gates than the circuit with the minimal number of Toffoli gates.

\begin{theorem}
    \label{thm:cnotcontrolled}
    The controlled version of any CNOT circuit on $n$ qubits can be implemented using $n - c$ Toffoli gates, where $c \geq 0$ is the number of generalized Jordan blocks of its parity matrix with polynomial $f(x) = (x + 1)^d$. This procedure uses at most $\frac{3}{2} \times$ the optimum number of Toffoli gates.
\end{theorem}

\begin{proof}
    Let the $n \times n$ parity matrix of the circuit be $A$. Construct the rational canonical form of $A$ as:
    $$
    A = S(C_{f_1} \oplus C_{f_2} \oplus \cdots \oplus C_{f_k})S^{-1}
    $$
    Using Gaussian elimination, synthesize CNOT circuits for $S$ and $S^{-1}$. Then the controlled circuit can be constructed as follows, where each $c(C_{f_i})$ is constructed according to Lemma \ref{thm:controlledcompanionmatrix} when $f_i \neq (x + 1)^{d_i}$ and using Lemma \ref{thm:controlledxp1companion} otherwise. This is possible because each $f_i$ must have $f_i(0) = 1$, otherwise $C_{f_i}$ would have an all-zero row, and the matrix would not be invertible, which is a contradiction since $A$ is a parity matrix and must be invertible.
    \ctikzfig{cnot-controlled}
    Since $C_{f_1} \oplus \cdots C_{f_k}$ has the same dimensions as $A$, the sum of the degrees $d_i$ of each $f_i$ is $n$. Each $c(C(_{fi}))$ requires $d_i$ Toffoli gates, except when $f_i(x) = (x + 1)^{d_i}$, when only $d_i - 1$ are required. Therefore, the total number of Toffoli gates will be $n - c$ where $c$ is the number of $f_i$ with this form. It is shown in Theorem \ref{thm:cnotoptimal} (Appendix \ref{app:asymptoticproof}) that this is approximately optimal.
\end{proof}

\subsection{Reducing the Toffoli-depth to $O(1)$ without ancillas}
\label{subsec:cnotdepthconstant}

We will now show that, at the cost of doubling the Toffoli count, the Toffoli depth of this construction can be reduced to $O(1)$ by using a technique known as \emph{toggle detection} \cite{gidney2015constructing}. This requires the use of borrowed ancillas -- these are ancilla qubits that are allowed to be in any state (not necessarily $\ket{0}$) and which must be returned to this state at the end of the computation. Crucially, for any operation that does not act on all the qubits in the circuit, the remaining qubits may be used as borrowed ancillas.

\begin{lemma}
    \label{thm:togglehermitian}
    For an arbitrary Hermitian unitary $V$, let $c(\bigotimes_{i=1}^nV_i)$ be $n$ copies of $c(V)$ sharing the same control and distinct targets. This unitary can be implemented in $c(V)$-depth of two using $n$ borrowed ancillas or $c(V)$ depth of four without ancillas. In both cases, the $c(V)$-count is $2n$.
\end{lemma}

\begin{proof}
    First, consider the case of $n = 1$. We have the following identity
    \ctikzfig{toggle-detection-single}
    since, if $V$ is Hermitian unitary then $V^2 = I$ and so the action on the middle register given the first and last qubits are in the states $\ket{x_1}$ and $\ket{x_2}$ is $V^{x_2}V^{x_1 \oplus x_2} = V^{2(x_2 - x_1x_2) + x_1} = V^{x_1}$. Now, for $n > 1$, note that by applying the $n = 1$ case to each copy of $V$ independently using a distinct borrowed ancilla, we find
    \ctikzfig{toggle-detection-many}
    and the $c(V)$ gates on the right-hand side can be organized into two layers with $c(V)$-depth one, indicated by the gray boxes. If no borrowed ancillas are available, we can split the computation in half: perform one half of the $c(V)$ gates in $c(V)$-depth two by borrowing qubits from the other half, and vice-versa, for a total $c(V)$-depth of four.
\end{proof}

\begin{lemma}
    \label{thm:shiftdepth}
    A controlled cyclic shift of $n$ qubits may be implemented with Toffoli-depth 8 and Toffoli-count $2(n - 1)$ without ancillas. 
\end{lemma}

\begin{proof}
    We will show that a cyclic shift $S$ can be implemented as a circuit of SWAP gates with depth two, and hence by applying Lemma \ref{thm:togglehermitian} $c(S)$ can be implemented with Toffoli-depth $8$. The cyclic shift $S$ can be implemented by first reversing the first $n - 1$ qubits, and then reversing all $n$ qubits -- this is illustrated below, where $R$ represents the permutation reversing a given set of qubits:
    \ctikzfig{cycle-reverse}
    To show that this is correct, consider qubit with index $i \neq n - 1$. This will be mapped to qubit $(n - 2) - i$ by the first reversal and then to $(n - 1) - (n - 2 - i) = i + 1$ by the second reversal. Now considering $i = n - 1$, this will be mapped to $(n - 1) - (n - 1) = 0$, and so the circuit correctly maps every qubit. The number of SWAP gates used for this representation is $n - 1$, and so the number of Toffoli gates needed after applying Lemma \ref{thm:togglehermitian} will be exactly $2(n - 1)$.
\end{proof}

\begin{theorem}
    \label{thm:ctrlcnotlowdepth}
    The controlled version of any CNOT circuit $C$ on $n$ qubits can be implemented with Toffoli-depth at most 12 and a Toffoli-count of at most $2n$, without ancillas.
\end{theorem}

\begin{proof}
    Using the same construction as in Theorem \ref{thm:cnotcontrolled}, write $c(C)$ as sequence of blocks $c(C_{f_i})$ which are disjoint except for the control qubit, up to a change of basis. For each $c(C_{f_i})$ with $f_i \neq 1$ (since $C_1$ is the identity), decompose it, as in Lemma \ref{thm:controlledcompanionmatrix}, into a controlled cyclic shift and a final Toffoli gate conjugated by CNOTs. All of the controlled cyclic shifts may be implemented simultaneously in Toffoli-depth 8 by applying Lemma \ref{thm:shiftdepth} to each. All of the final Toffoli gates can be implemented simultaneously in Toffoli-depth 4 by applying Lemma \ref{thm:togglehermitian}. The Toffoli-count is more than doubled relative to Lemma \ref{thm:controlledcompanionmatrix} due to the application of Lemma \ref{thm:togglehermitian} throughout; in Theorem \ref{thm:depthoptimalapprox} it is shown that this is approximately optimal within a factor of $6$.
\end{proof}

\subsection{Reducing Toffoli-depth to one with $O(n)$ ancillas}
\label{sec:cnot_toffoli_one}

Now we describe an alternative method to implement a controlled CNOT circuit. Compared to the method of Section~\ref{subsec:cnotdepthconstant}, there are two main differences. First, this method achieves the optimal Toffoli depth, which is $1$. Second, it uses a measurement-based uncomputation, which involves applying a Clifford correction based on a measurement~\cite{Gidney2018halvingcostof} - the uncomputation can be realized in depth of one, and as such, the reaction depth of this approach is two. Note that this also implies that a $T$-depth of $1$ is achievable with $O(n)$ ancillas, since any Toffoli gate can be implemented with a $T$-depth of one using $O(1)$ ancilla qubits \cite{Selinger2013}.

\begin{theorem}
    \label{thm:cnotoptimaldepth}
    The controlled version of any CNOT circuit on $n$ qubits can be implemented with Toffoli-depth exactly $1$ and a Toffoli-count of at most $n$, using $2n - 1$ clean ancillas.
\end{theorem}

\begin{proof}
    A high-level description of this circuit is shown in Figure~\ref{fig:depth_one}, which we will explain in more detail now. Without loss of generality, let $A$ be invertible matrix associated with a CNOT circuit $U$ (as in Equation ~\eqref{eq:cnot_circuit_binary_matrix}). We aim to implement $c(U)$. To that end, let $D=A+I$, and note the following identity:
    \begin{equation}
        c(U)(|a\rangle|\vec{x}\rangle) = |a\rangle|\vec{x} + aD\vec{x}\rangle.  \label{eq:cU_based_on_D}  
    \end{equation}
    We will decompose the right-hand side of Equation \eqref{eq:cU_based_on_D} into more elementary operations. For this purpose, the following two unitaries will be useful. 
    \begin{align*}
        U_D(|\vec{x}\rangle|\vec{y}\rangle) &= |\vec{x}\rangle|\vec{y} + D\vec{x}\rangle, \\
        \overline{\text{CNOT}} (|\vec{x}\rangle|\vec{y}\rangle) &= |\vec{x} + \vec{y}\rangle |\vec{y}\rangle. 
    \end{align*}
    Note that both unitaries can be represented in terms of CNOT circuits because they are linear and invertible on the bitstrings. Given $n$ ancillas, by applying $U_D$ and then $c(\overline{\text{CNOT}})$, we obtain the following state:
    \begin{equation}
            c(\overline{\text{CNOT}})(I\otimes U_D)(|a\rangle|\vec{x}\rangle|0^n\rangle )
            = |a\rangle|\vec{x} + aD\vec{x}\rangle|D\vec{x}\rangle.
        \label{eq:cU_D_unitary_part}
    \end{equation}
    Note the only non-Clifford part of this circuit is $c(\overline{\text{CNOT}})$, which consists of at most $n$ Toffolis. Moreover, these Toffolis can be furthermore parallelized to a depth-$1$ circuit using fanout. 
    
    Equation~\eqref{eq:cU_D_unitary_part} is almost what we need, but not quite. In order to apply Equation~\eqref{eq:cU_based_on_D}, one needs to uncompute the third register. Happily, the uncomputation can be done by measuring the third register in the $X$-basis and applying CZ and $Z$ corrections that depend on the measurement outcome. Without loss of generality, let $\vec{m}$ be the measurement outcome. The resulting state, up to normalization, is $(-1)^{D\vec{x} \cdot \vec{m}}|a\rangle |A^a\vec{x}\rangle$. Denoting $\vec{y}=A^a\vec{x}$, we get
    \begin{align*}
        D\vec{y} + a(A + A^{-1})\vec{y} &= DA^a\vec{x} + a(A + A^{-1})A^a\vec{x} \\
            &= DA^a\vec{x} + a(A + A^{-1})A\vec{x} \\
            &= D[(1 - a)I + aA]\vec{x} + a(I + A^2)\vec{x} \\
            &= [I + A + a(I + A^2)]\vec{x} + a(I + A^2)\vec{x} = D\vec{x}
    \end{align*}
    Therefore, the uncomputation can be carried out by applying the following phase 
    \begin{equation*}
        |a\rangle|y\rangle \to (-1)^{\vec{y}\cdot \vec{u} + a\vec{y} \cdot \vec{v}} |a\rangle|y\rangle,
    \end{equation*}
    where $\vec{u} = D^T\vec{m}$ and $\vec{v}= (A+A^{-1})^T\vec{m}$. Furthermore,
    \begin{equation*}
        \begin{aligned}
            \left(I\otimes \prod_{k=1}^{n} Z_k^{u_k} \right) (|a\rangle|y\rangle) &= (-1)^{\vec{y} \cdot \vec{u}} |a\rangle|y\rangle, \\
            c\left(\prod_{k=1}^{n} Z_{k}^{v_k}\right) (|a\rangle |y\rangle) = \prod_{k=1}^{n} \text{CZ}_{a,k}^{v_k} (|a\rangle |y\rangle) &= (-1)^{a\vec{y}\cdot \vec{v}}|a\rangle|y\rangle,
        \end{aligned}
    \end{equation*}
    where $Z_k$ is the Pauli $Z$ on the $k$'th qubit of the second register and $\text{CZ}_{a,k}$ is the CZ gate between the control qubit and the $k$'th qubit of the second register. The $k$'th component of the vectors $\vec{u}$ and $\vec{v}$ are denoted as $u_k$ and $v_k$ respectively. Therefore, the uncomputation can be carried out by a set of $Z$ and CZ gates, which can can be parallelized by applying a fanout to the control qubit. Theorem \ref{thm:depthoptimalapprox} shows that a small addition to this procedure results in an approximately optimal Toffoli count within a factor of $3$, while keeping the $T$-depth and reaction depth the same.
\end{proof}

\begin{figure}
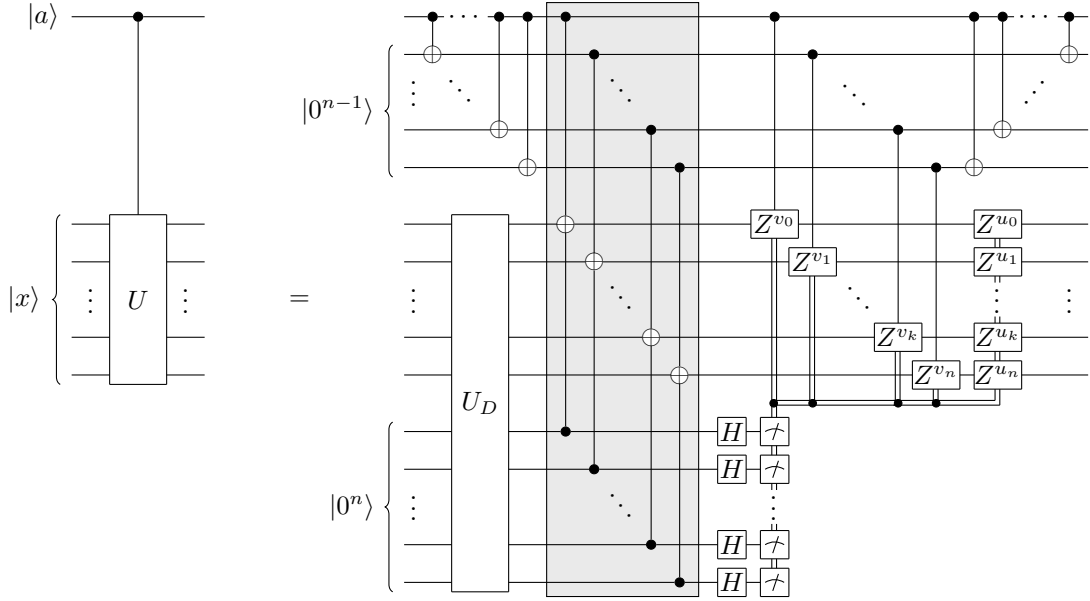

    \centering
    \ctikzfig{controlled-cnot-depth-1}
    \caption{A high-level overview of the circuit that implements $c(U)$ with Toffoli-depth $1$, where $U$ is a CNOT circuit. Here, $U_D$ implements $|a\rangle|\vec{x}\rangle|0\rangle \to |a\rangle|\vec{x}\rangle|D\vec{x}\rangle$ while $\vec{u}$ and $\vec{v}$ are binary vectors that depend on the measurement outcome. The only non-Clifford part of this circuit is the $n$ Toffolis, highlighted in grey, which are trivially parallelizable.}
    \label{fig:depth_one}
\end{figure}

\section{Generalizations}
\label{sec:gen}

\subsection{Controlled Clifford Circuits}
\label{subsec:controlledclifford}

We start by noting that a controlled depth-$1$ layer of $S$ gates can be implemented in constant $T$-depth whenever the whole circuit contains at least three qubits. This restriction is necessary because $c(S)$ cannot be synthesized in the Clifford+T gate set using fewer than two borrowed ancillas \cite[Corollary 2]{giles2013exact}.

\begin{lemma}
    \label{thm:ctrlsdepth}
    $c(S^{\otimes k})$ can be implemented in constant $T$-depth and $O(k)$ $T$-count whenever $k \geq 3$ or two borrowed ancillas are available.
\end{lemma}

\begin{proof}
    We have the following generalization of toggle-detection for $c(S)$, which demonstrates that $c(S)$ can be implemented\footnote{This follows from a similar argument to Lemma \ref{thm:togglehermitian}, except that the $V^{2(x_2 - x_1x_2)}$ term which cannot be eliminated when $V^2 \neq I$ is instead expanded recursively. To verify its correctness, see the following implementation using Quirk: \href{https://tlaakkonen.github.io/static/quirk-ctrls.html}{https://tlaakkonen.github.io/static/quirk-ctrls.html}} in constant $T$-depth with two borrowed ancillas:
    \ctikzfig{toggle-detection-controls}
    Similar to Lemma \ref{thm:togglehermitian}, split the set of target qubits into three groups. For each group, apply the above identity to each $c(S)$ operation, borrowing two independent ancillas from the remaining qubits, one from each of the non-participating groups. Each group can be performed in constant $T$-depth, since the borrowed and target qubits are all distinct, and the only interaction with the control qubit is via Clifford gates. Since a Toffoli gate can be implemented with $7$ $T$-gates in $T$-depth $3$ without ancillas \cite{Amy2013Meet}, the $T$-depth is at most $54$ and the $T$-count is $42k$.
\end{proof}

Combined with the canonical form of \cite{aaronson2004improved}, this is enough to construct controlled Clifford circuits with $O(n)$ $T$-count and constant $T$-depth.

\begin{theorem}
    \label{thm:ctrlclifford}
    Given a Clifford circuit $U$ on $n \geq 3$ qubits, $c(U)$ can be constructed with constant $T$-depth and $O(n)$ $T$-count, without ancillas.
\end{theorem}

\begin{proof}
    In \cite[Theorem 8]{aaronson2004improved}, it is shown that every Clifford circuit can be written as a sequence of layers H-C-P-C-P-C-H-P-C-P-C, where each layer contains only gates of that type (H corresponding to $H$ gates, P corresponding to $S$ gates, and C corresponding to CNOT gates). We will handle each type of layer separately:
    \begin{itemize}
        \item For $H$ gates, since $H^2 = I$ such a layer can be collapsed to depth $1$. Each $c(H)$ gate can be implemented using $2$ $T$-gates, in a way that is trivially parallelized to a $T$-depth of $2$, since the only interaction with the control qubit is via a Clifford gate:
        $$\begin{tikzpicture}
	\begin{pgfonlayer}{nodelayer}
		\node [style=none] (0) at (0, 0) {};
		\node [style=none] (1) at (2, 0) {};
		\node [style=gate] (2) at (1, -1) {$H$};
		\node [style=none] (3) at (0, -1) {};
		\node [style=none] (4) at (2, -1) {};
		\node [style=cnot ctrl] (5) at (1, 0) {};
		\node [style=none] (6) at (4, 0) {};
		\node [style=none] (7) at (14.125, 0) {};
		\node [style=none] (9) at (4, -1) {};
		\node [style=none] (10) at (14.125, -1) {};
		\node [style=none] (12) at (3, -0.5) {$=$};
		\node [style=cnot targ] (15) at (9.25, -1) {};
		\node [style=cnot ctrl] (16) at (9.25, 0) {};
		\node [style=gate] (19) at (6, -1) {$S^\dagger$};
		\node [style=gate] (20) at (7.125, -1) {$H$};
		\node [style=gate] (21) at (4.875, -1) {$H$};
		\node [style=gate] (22) at (8.25, -1) {$T^\dagger$};
		\node [style=gate] (24) at (10.25, -1) {$T$};
		\node [style=gate] (26) at (12.375, -1) {$S$};
		\node [style=gate] (27) at (13.375, -1) {$H$};
		\node [style=gate] (28) at (11.375, -1) {$H$};
	\end{pgfonlayer}
	\begin{pgfonlayer}{edgelayer}
		\draw (5) to (2);
		\draw (2) to (3.center);
		\draw (2) to (4.center);
		\draw (1.center) to (5);
		\draw (5) to (0.center);
		\draw (16) to (15);
		\draw (10.center) to (27);
		\draw (27) to (26);
		\draw (26) to (28);
		\draw (28) to (24);
		\draw (24) to (15);
		\draw (15) to (22);
		\draw (20) to (19);
		\draw (19) to (21);
		\draw (21) to (9.center);
		\draw (22) to (20);
		\draw (7.center) to (16);
		\draw (16) to (6.center);
	\end{pgfonlayer}
\end{tikzpicture}$$
        \item For $S$ gates, since $S^2 = Z$ such a layer can be written as a depth-$1$ layer of $S$ followed by a layer of $Z$. Controlling the $Z$ layer is entirely Clifford, and controlling the $S$ layer follows from Lemma \ref{thm:ctrlsdepth}.
        \item For CNOT gates, this follows directly from Theorem \ref{thm:ctrlcnotlowdepth}.
    \end{itemize}
    Each of these layers can be controlled in constant $T$-depth and $O(n)$ $T$-count, and hence the whole circuit. In particular, we have $T$-depth at most $280$ and $T$-count at most $182n$.
\end{proof}

\subsection{Controlling unitaries with constant $T$-depth}
\label{subsec:controlledunitary}

To tackle Clifford+T circuits, we can generalize Lemma \ref{thm:ctrlsdepth} to a controlled layer of $T$ gates. This requires the circuit to contain at least four qubits, because $c(T)$ cannot be synthesized in the Clifford+T gate set using fewer than three borrowed ancillas, for essentially the same reason as the above case \cite[Corollary 2]{giles2013exact}. 

\begin{lemma}
    \label{thm:ctrltdepth}
    $c(T^{\otimes k})$ can be implemented in constant $T$-depth and $O(k)$ $T$-count whenever $k \geq 4$ or three borrowed ancillas are available.
\end{lemma}

\begin{proof}
    The following identity demonstrates that $c(T)$ can be implemented\footnote{This was derived similarly to Lemma \ref{thm:ctrlsdepth}. To verify its correctness, see the following implementation using Quirk: \href{https://tlaakkonen.github.io/static/quirk-ctrlt.html}{https://tlaakkonen.github.io/static/quirk-ctrlt.html}} in constant $T$-depth with three borrowed ancillas:
    \ctikzfig{toggle-detection-controlt2}
    Similarly to Lemma \ref{thm:ctrlsdepth}, split the set of target qubits into four groups, and apply the above identity to each $c(T)$ operation, borrowing three qubits each from the remaining groups. This has $T$-depth at most $144$ and $T$-count $78k$.
\end{proof}

Given this, we can construct the controlled version of any Clifford+T unitary $U$ with only a constant multiplicative factor overhead in $T$-depth and a $T$-count overhead that is linear in the $T$-count of the original unitary and the number of qubits. This is approximately optimal in terms of $T$-depth (assuming an optimal input circuit), since any circuit for $c(U)$ can be converted into a circuit for $U$ by taking the control qubit to be an ancilla and fixing its state as $\ket{1}$, but may not be optimal in terms of $T$-count.

\begin{theorem}
    \label{thm:cliffordt}
    Given a unitary $U$ on $n \geq 4$ qubits with $T$-depth $D$ and $T$-count $C$, $c(U)$ can be constructed with $T$-depth $O(D)$ and $T$-count $O(C + n)$, without ancillas.
\end{theorem}

\begin{proof}
    By definition, $U$ can be written as a sequence of layers $U = C_DT^{\otimes k_D}\cdots C_1T^{\otimes k_1}C_0$ where each $C_i$ is a Clifford circuit and $\sum_i k_i = C$. Now, supposing we replace each $T^{\otimes k_i}$ with $c(T^{\otimes k_i})$, then whenever the control is disabled, the circuit will become $V = C_D\cdots C_1C_0$, which is Clifford. Hence by adding an anti-controlled $c(V)$ operation, $c(U)$ can be written as follows:
    \ctikzfig{controlled-unitary}
    Since $c(V)$ can be implemented in constant $T$-depth and $T$-count $O(n)$ from Theorem \ref{thm:ctrlclifford}, and each $c(T^{\otimes k_i})$ can be implemented in constant $T$-depth and $T$-count $O(k_i)$ from Lemma \ref{thm:ctrltdepth}, the overall circuit has $T$-depth at most $144D + 280$ and $T$-count at most $78C + 182n$.
\end{proof}

The reason the constant factors in Theorems \ref{thm:ctrlclifford} and \ref{thm:cliffordt} are so large is because Lemmas \ref{thm:ctrlsdepth} and \ref{thm:ctrltdepth} do not use any ancilla qubits. Given access to $O(k)$ ancillas, we can parallelize $c(S^{\otimes k})$ operations much more efficiently:
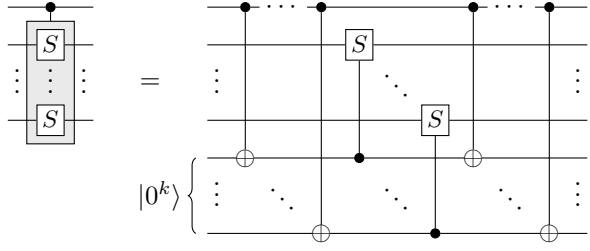
\ctikzfig{ctrls-par-ancillas}
This has $T$-count $3k$ and $T$-depth $2$. Combining this with Theorem \ref{thm:cnotoptimaldepth} and applying the $T$-count 4 and $T$-depth 1 construction of the Toffoli gate with two ancillas due to Jones \cite{Jones2013} would improve Theorem \ref{thm:ctrlclifford} to $T$-depth at most $17$ and $T$-count at most $36n$. Similarly for $c(T^{\otimes k})$ we have
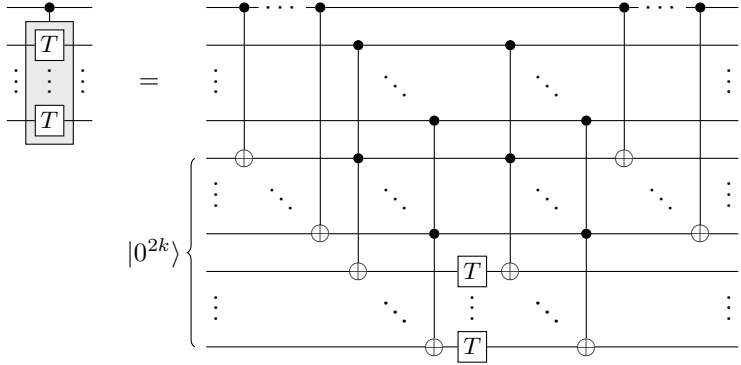
\ctikzfig{ctrlt-par-ancillas}
with $T$-count $9k$ and $T$-depth $3$; this would improve Theorem \ref{thm:cliffordt} to $T$-depth at most $3D + 17$ and $T$-count at most $9C + 36n$.

\section{Applications}
\label{sec:applications}

\subsection{Catalyzing arbitrary-angle rotations}
\label{sec:catalysis}

There are two principal approaches to synthesizing arbitrary-angle rotations in the Clifford+T gate set: those based on the Solovay-Kitaev theorem \cite{dawson2006solovay}, and those based on catalysis or phase kickback \cite{kim2025catalytic,Gidney2018halvingcostof}. Here, extending the method of \cite{kim2025catalytic}, we will show that our constant $T$-depth construction for controlled CNOT circuits can be used to improve the catalysis-based approach, leading to the following result: a single-qubit rotation by any angle $\alpha$ can be catalyzed to precision $\epsilon$ with $T$-depth exactly $1$, using a catalyst state on $O(\log(\frac{1}{\epsilon}))$ qubits. Moreover, this catalyst state is \emph{universal} in the sense that it does not depend on $\alpha$. This construction is therefore an improvement both on \cite{kim2025catalytic}, which needs a $O(\log(\frac{1}{\epsilon})^2)$-qubit catalyst state when it is required to be universal, and arithmetic-based methods \cite{Gidney2018halvingcostof}, which have a universal $O(\log(\frac{1}{\epsilon}))$-qubit catalyst state but require a $T$-depth of $O(\log\log(\frac{1}{\epsilon}))$.

Given a precision $\epsilon \leq 1$, let us pick $n$ such that $\frac{2\pi}{2^n - 1} < \epsilon$. For $n > \log_2\frac{4}{4-\pi} \approx 2.2$ we have $\frac{2\pi}{2^n - 1} < \frac{8}{2^n}$ and so $n = \lceil\log_2\frac{8}{\epsilon}\rceil \geq 3$ suffices. Then pick a irreducible polynomial $f(x) \in \mathbb{F}_2[x]$ with degree $n$ such that $x$ is a primitive element in $\mathrm{GF}(2^n) \equiv \mathbb{F}_2[x] / \langle f(x) \rangle$ (this is always possible). Let $C_f$ be the companion matrix of this polynomial over $\mathbb{F}_2$, and $U_f$ the unitary corresponding to the CNOT circuit with parity matrix $C_f$. It is shown in \cite{kim2025catalytic} that $U_f$ has a complete set of eigenvectors given by
$$ \ket{\psi_0} = \ket{0 \cdots 0}, \qquad \ket{\psi_k} = \frac{1}{\sqrt{2^n - 1}} \sum_{j = 0}^{2^n - 2} \tilde{\omega}^{-jk} \ket{C_f^jv} \text{ for } k \in \{1, \dots, 2^n - 1\}$$
where $v = \begin{pmatrix} 1 & 0 & \cdots & 0 \end{pmatrix}^T$, $\tilde{\omega} = e^{\frac{2\pi i}{2^n - 1}}$, and the eigenvalue corresponding to $\ket{\psi_k}$ is $\tilde{\omega}^k$. Now for any angle $\alpha \in [0, 2\pi)$ and integer $k$ which is coprime to $2^n - 1$, we can find $d, d' \in \{0, \dots, 2^n - 2\}$ such that $|\alpha - \frac{2 \pi d}{2^n - 1}| < \epsilon$ and $d'k \equiv d \pmod{2^n - 1}$. Let $U_{f,d'}$ be the unitary corresponding to a CNOT circuit with parity matrix $C_f^{d'}$, then we must have $U_{f,d'} = U_f^{d'}$, and so:
\begin{align*}
    c(U_{f,d'}) [(a \ket{0} + b \ket{1})~\ket{\psi_k}] &= a\ket{0}\ket{\psi_k} + b\ket{1}U_{f,d'}\ket{\psi_k} \\
        &= a\ket{0}\ket{\psi_k} + be^{i\frac{2\pi d'k}{2^n - 1}}\ket{1}\ket{\psi_k}\\
        &= a\ket{0}\ket{\psi_k} + be^{i\frac{2\pi d}{2^n - 1}}\ket{1}\ket{\psi_k}\\
        &\approx a\ket{0}\ket{\psi_k} + be^{i\alpha}\ket{1}\ket{\psi_k} = (R_Z(\alpha) \otimes I) [(a \ket{0} + b \ket{1})~\ket{\psi_k}]
\end{align*}
Therefore, $c(U_{f,d'})$ catalyzes the rotation $R_Z(\alpha)$ up to precision $\epsilon$ on its control qubit, using the universal catalyst state $\ket{\psi_k}$ (which depends only on $\epsilon$, not $\alpha$). Moreover, as $U_{f,d'}$ is a CNOT circuit, $c(U_{f,d'})$ can be implemented with $T$-depth exactly $1$ by Theorem \ref{thm:cnotoptimaldepth}. By including the required $O(n)$ ancilla qubits in the catalyst state, we obtain an $O(\log(\frac{1}{\epsilon}))$-qubit catalyst state overall. The circuits implementing these operations can be determined classically in $O(\mathrm{polylog}(\frac{1}{\epsilon}))$ time by applying exponentiation by squaring to $C_f$. 

To prepare a catalyst state, note that, as discussed in \cite{kim2025catalytic}, applying quantum phase estimation to $U_f$ from the starting vector $\ket{10\cdots 0}$ will generate a state $\ket{\psi_{k'}}$ with uniformly random $k'$. This procedure has a $T$-depth of $\tilde{O}(\log\frac{1}{\epsilon})$ by applying Theorem \ref{thm:cnotoptimaldepth} to each $c(U_f^{2^t})$ operation, where $\tilde{O}(.)$ hides $\log\log(\frac{1}{\epsilon})$ factors that are due to the quantum Fourier transform. This procedure can be repeated until $k'$ is coprime to $2^n - 1$; on average, this will take $O(\log\log(\frac{1}{\epsilon}))$ iterations. In \cite{kim2025catalytic}, this was then further refined to generate a specific $\ket{\psi_k}$ with $k$ fixed using a coherent form of Shor's algorithm. However, since our construction above is agnostic to the choice of $k$, so long as it is coprime to $2^n - 1$, we can defer the compilation of the rest of the circuit until after $k$ has been picked randomly (in effect, producing a dynamic circuit). The $T$-depth required to prepare the catalyst state is thus $\tilde{O}(\log(\frac{1}{\epsilon}))$ on average, and this implies that a $q$-qubit Clifford+$R_Z$ circuit containing $r$ $Z$-rotation gates can be compiled with precision $\epsilon$ to a dynamic Clifford+T circuit with $T$-depth $r + \tilde{O}(\log\frac{r}{\epsilon})$ and $q + O(\log\frac{r}{\epsilon})$ qubits. 

\subsection{Variations on the swap test}

The Hadamard test is a simple quantum algorithm that measures $\mathrm{Tr}(\rho O)$ for unitary $O$ by a single application of $c(O)$. Thus, if $O$ is Clifford, Theorems \ref{thm:cnotoptimaldepth} and \ref{thm:ctrlclifford} show that this can be done in constant $T$-depth, or $T$-depth exactly $1$ for a CNOT circuit. A special case of the Hadamard test is the swap test, which measures the overlap between two quantum states by setting $O = \mathrm{SWAP}$, noting that $\mathrm{Tr}((\rho \otimes \sigma) \cdot \mathrm{SWAP}) = \mathrm{Tr}(\rho\sigma)$. This can be generalized \cite{Quek2024} to compute $\mathrm{Tr}(\rho_1\rho_2\cdots \rho_N)$, as well as other quantities such as Shatten $p$-norms, by setting $O$ to be a cyclic shift of qubits; this is also used as a subroutine for parallelizing quantum signal processing \cite{Martyn2025}. An implementation with overall constant depth and a $T$-depth of $2$ was given in \cite[Section 3]{Quek2024}, and Theorem \ref{thm:cnotoptimal} implies that this is approximately optimal in terms of Toffoli-count. However, since any permutation of qubits can be realized as a CNOT circuit, Theorem \ref{thm:cnotoptimaldepth} shows that this can be improved to a $T$-depth of $1$.

\section{Discussion}

In this paper, we have shown how to construct the controlled versions of CNOT and Clifford circuits with approximately optimal non-Clifford depths and gate counts. Using this, we showed that controlled Clifford+T circuits can be constructed with only a constant-factor $T$-depth overhead relative to their uncontrolled counterparts. In fact, we can draw a similar conclusion for Clifford+$R_Z$ circuits, since in Section \ref{sec:catalysis} we use our constructions to show that arbitrary-angle $Z$-rotations can be implemented approximately in $T$-depth exactly $1$ given access to a universal catalyst state. Two open questions are given in \cite{kim2025catalytic} regarding constant $T$-depth catalysis, the first of which is whether the size of the catalyst state can be reduced from $O(\log(\frac{1}{\epsilon})^2)$ to $O(\log\frac{1}{\epsilon})$ where $\epsilon$ is the desired precision. Our construction resolves this question positively. 

The second open question of \cite{kim2025catalytic} is as follows: can rotation angles of the form $\frac{2\pi}{2^k}$ be catalyzed \emph{exactly} in constant $T$-depth with a catalyst state of size $O(k)$? These angles are desirable because they are used directly in constructions of the quantum Fourier transform and multi-qubit Toffoli gates. Unfortunately, our approach cannot tackle this: to achieve angles of the form $\frac{2\pi}{q}$, we would require the existence of $A \in GL_n(\mathbb{F}_2)$ with order $q$, from which we could derive an $n$-qubit catalyst state. In the case of $q = 2^k$, we are forced\footnote{Supposing that $q = 2^k$, let $\mu_A(x)$ be the minimal polynomial of $A$. Since $A^q = A^{2^k} = I$, then $\mu_A(x) \mid x^{2^k} - 1$. But as $x^{2^k} - 1 \equiv (x + 1)^{2^k} \pmod{2}$ we must have $\mu_A(x) = (x + 1)^d$ for some $d \leq n$, and so $\mu_A(x) \mid (x + 1)^{2^{\lceil \log_2 n \rceil}}$. Finally since $(x + 1)^{2^{\lceil \log_2 n \rceil}} \equiv x^{2^{\lceil \log_2 n \rceil}} - 1 \pmod{2}$, then $A^{2^{\lceil \log_2 n \rceil}} = I$, and so the order $q = 2^k$ must divide $2^{\lceil \log_2 n \rceil}$, which implies $k \leq \lceil \log_2 n \rceil$.} to have $n \geq \frac{q}{2}$, which is much larger than $O(k) = O(\log q)$. Furthermore, moving from controlled CNOT circuits to controlled Clifford circuits does not improve matters, since $n$-qubit Clifford circuits correspond \cite{Dehaene2003} (up to phases) to elements of the symplectic group $Sp(2n, \mathbb{F}_2) \subset GL_{2n}(\mathbb{F}_2)$, which would lead to $n \geq \frac{q}{4}$.

While we have shown an exactly optimal construction for one parameter (the $T$-depth of controlled CNOT circuits using $O(n)$ ancillas), substantial improvement may be possible in our other constructions. In particular, considering alternative Clifford normal forms may lead to a better construction for controlled Clifford circuits, and a more careful accounting of borrowed ancillas may allow for better constant factors in the ancilla-free constructions.

\subsection*{Acknowledgements}

TL acknowledges support from the U.S. Department of Energy under Award No. DE-SC0020264, and thanks John van de Wetering for advice on lower bounds for non-Clifford operations. IK acknowledges support from the Advanced Scientific Computing Research program in the Office of Science of the Department of Energy (DE-SC0026109).

\bibliographystyle{quantum}
\bibliography{refs}

\appendix

\section{Proof that Theorem \ref{thm:cnotcontrolled} is Approximately Optimal}
\label{app:asymptoticproof}

To show that Theorem \ref{thm:cnotcontrolled} is approximately optimal, we will use the concept of \emph{unitary stabilizer nullity} \cite{jiang2023lower}, which is defined in terms of the unitary Pauli function $P_U(A, B)$. We will call the Pauli strings $A$ for which there exists a $B$ such that $|P_U(A, B)| = 1$ the generalized stabilizers of $U$. This function is given for an $n$-qubit unitary $U$ and Pauli strings $A$ and $B$ as:
$$ P_U(A, B) = \frac{1}{2^n} \trace(A U B U^\dagger) $$
Throughout, we define the Pauli matrices $A_i$ and $B_i$ by
$$ A = \bigotimes_{i=1}^n A_i~, \quad B= \bigotimes_{i=1}^n B_i$$
as well as $A_{[2,n]} = \bigotimes_{i=2}^n A_i$ so that $A = A_1 \otimes A_{[2,n]}$, and likewise $B = B_1 \otimes B_{[2,n]}$. We will now fully characterize $P_{c(C)}$, and begin by showing that the first qubit of any generalized stabilizer must be $I$ or $Z$:
\begin{lemma}
    \label{thm:noxstabs}
    Given a CNOT circuit $C \neq I$, we have $|P_{c(C)}(A, B)| < 1$ if $B_1 \in \{X, Y\}$.
\end{lemma}

\begin{proof}
    Let $B_1 = Z^bX$. Consider the case when $A_1 = Z^a$. Then:
    \begin{align*}
        \trace[(Z_1^aA_{[2,n]})c(C)&(Z_1^bX_1B_{[2,n]})c(C)^\dagger] \\
        &= \begin{aligned}[t]&\trace[(\bra{0} \otimes A_{[2,n]}) c(C) (Z_1^bX_1B_{[2,n]}) c(C)^\dagger (\ket{0} \otimes I)] \\ &+ (-1)^a\trace[(\bra{1} \otimes A_{[2,n]}) c(C) (Z_1^bX_1B_{[2,n]}) c(C)^\dagger (\ket{1} \otimes I)] \end{aligned}\\
        &= \begin{aligned}[t]&\trace[(\bra{1} \otimes A_{[2,n]}B_{[2,n]})(\ket{0} \otimes I)] \\ &+ (-1)^{a+b}\trace[(\bra{0} \otimes A_{[2,n]}CB_{[2,n]})(\ket{1} \otimes C^\dagger)] \end{aligned}\\
        &= \braket{0}{1} \trace[A_{[2,n]}B_{[2,n]}] + \braket{1}{0}(-1)^{a+b}\trace[A_{[2,n]}CB_{[2,n]}C^\dagger] = 0
    \end{align*}
    Hence, $|P_{c(C)}(A, B)| = 0 < 1$. Now consider the case when $A_1 = Z^aX$. Then:
    \begin{align*}
        \trace[(Z_1^aX_1A_{[2,n]})c(C)&(i^bZ_1^bX_1B_{[2,n]})c(C)^\dagger] \\ &= \begin{aligned}[t]&\trace[(\bra{1} \otimes A_{[2,n]}) c(C) (Z_1^bX_1B_{[2,n]}) c(C)^\dagger (\ket{0} \otimes I)] \\ &+ (-1)^a\trace[(\bra{0} \otimes A_{[2,n]}) c(C) (Z_1^bX_1B_{[2,n]}) c(C)^\dagger (\ket{1} \otimes I)] \end{aligned} \\
        &= \begin{aligned}[t] &\trace[(\bra{0} \otimes A_{[2,n]}CB_{[2,n]})(\ket{0} \otimes I)] \\ &+ (-1)^{a+b}\trace[(\bra{1} \otimes A_{[2,n]}B_{[2,n]})(\ket{1} \otimes C^\dagger)] \end{aligned} \\
        &= \trace[A_{[2,n]}CB_{[2,n]}] + (-1)^{a+b}\trace[A_{[2,n]}B_{[2,n]}C^\dagger]
    \end{align*}
    Since $C$ is a CNOT circuit, it is a permutation matrix, and so $A_{[2,n]}CB_{[2,n]} = DC'$ for some diagonal matrix $D$ and a permutation matrix $C'$. In particular, if $C \neq I$ then we have $C' \neq I$, and hence $\trace[A_{[2,n]}CB_{[2,n]}] < 2^{n - 1}$.  Likewise, we must have $\trace[A_{[2,n]}B_{[2,n]}C^\dagger] < 2^{n - 1}$. Therefore, $|P_{c(C)}(A, B)| < 1$ from the triangle inequality.
\end{proof}

From \cite[Lemma 3]{jiang2023lower}, to determine the generalized stabilizers of $c(C)$, we need to find all those Pauli strings $A$ such that $c(C)Ac(C)^\dagger$ is again a Pauli string. We just showed that this cannot be the case if $A_1 \in \{X, Y\}$, and we will now use this to determine the total number of generalized stabilizers. In particular, for any Pauli string $P$ let $\vec{x}(P)$ and $\vec{z}(P)$ be binary vectors defined as follows
$$ \vec{x}(P)_i = \begin{cases} 1 & P_i \in \{X, Y\} \\ 0 & P_i \in \{I, Z\} \end{cases} \qquad \vec{z}(P)_i = \begin{cases} 1 & P_i \in \{Z, Y\} \\ 0 & P_i \in \{I, X\} \end{cases}$$
then we have:

\begin{lemma}
    \label{thm:eigenstabs}
    Given a CNOT circuit $C$, $c(C)$ has exactly $2^{2\lambda + 1}$ generalized stabilizers, where $\lambda$ is the number of linearly independent eigenvectors of the parity matrix of $C$ over $\mathbb{F}_2$.
\end{lemma}

\begin{proof}
    From Lemma \ref{thm:noxstabs}, suppose $A = Z^a\otimes A_{[2,n]}$ is a generalized stabilizer of $c(C)$. We can compute
    \begin{align*}
        c(C) (Z^a \otimes A_{[2,n]}) c(C)^\dagger &= (I \oplus C)(A_{[2,n]} \oplus (-1)^aA_{[2,n]}) (I \oplus C^\dagger) = A_{[2,n]} \oplus (-1)^aCA_{[2,n]}C^\dagger \\ &= (Z^a \otimes A_{[2,n]}) ~c(A_{[2,n]} CA_{[2,n]}C^\dagger)
    \end{align*}
    Clearly, $B_{[2,n]} = CA_{[2,n]}C^\dagger$ must be a Pauli string. Let $M_C$ be the parity matrix corresponding to $C$, then we have $\vec{x}(B_{[2,n]}) = M_C\vec{x}(A_{[2,n]})$ (essentially by definition), and $\vec{z}(B_{[2,n]}) = M_C^T\vec{z}(A_{[2,n]})$ (since conjugating any CNOT with Hadamard gates to change from the $X$ to $Z$ basis swaps the control and target qubits). 

    Since for any Pauli string $P$, $c(P)$ is never a Pauli string except in the trivial case of $P = I$, for $(Z^a \otimes A_{[2,n]})c(A_{[2,n]}B_{[2,n]})$ to be a Pauli string, we must have that $B_{[2,n]} = A_{[2,n]}$. This is accomplished exactly when $\vec{x}(A_{[2,n]})$ is an eigenvector of $M_C$ and $\vec{z}(A_{[2,n]})$ is an eigenvector of $M_C^T$ over $\mathbb{F}_2$ with eigenvalues one or zero (in which case it must be the zero vector, since $M_C$ is invertible). Note that since all eigenvectors over $\mathbb{F}_2$ must have an eigenvalue of zero or one, this covers all the eigenvectors of $M_C$.
    
    Finally, because every matrix is similar to its transpose, the number of linearly independent eigenvectors of $M_C$ must be the same as for $M_C^T$. Let this quantity be called $\lambda$. Then there are $(2^\lambda)^2$ pairs of eigenvectors, and each of these corresponds to a unique choice of $A_{[2,n]}$. We can then choose $a$ to be either $+1$ or $-1$, giving a total of $2^{2\lambda + 1}$ generalized stabilizers.
\end{proof}

We can connect this to the formalism in Theorem \ref{thm:cnotcontrolled} as follows:
\begin{lemma}
    \label{thm:eigenblocks}
    Given a parity matrix $M$, let $c$ be as defined in Theorem \ref{thm:cnotcontrolled}, then $c = \lambda$, where $\lambda$ is the number of linearly independent eigenvectors of $M$ over $\mathbb{F}_2$.
\end{lemma}

\begin{proof}
    Clearly, the eigenvectors of $M$ correspond directly to the eigenvectors of its generalized Jordan normal form up to a change of basis. Consider the normal form of $M$ block by block. First, note that any block $C_f$ must be invertible, since $M$ is invertible. Therefore, the only eigenvector with eigenvalue zero is the zero vector, which does not contribute to $\lambda$. To find the eigenvectors with eigenvalue one, let $C_f$ be a block with minimal polynomial $f(x)$. 
    
    Suppose that $(x + 1) \nmid f(x)$, then we must have $\gcd(x + 1, f(x)) = 1$ since $x + 1$ is irreducible. Therefore, by Bezout's lemma, there exists polynomials $a(x)$ and $b(x)$ such that $a(x)(x + 1) + b(x)f(x) = 1$, which implies that $a(C_f)(C_f + I) + b(C_f)f(C_f) = I$. But since $f(C_f) = 0$ by definition, we have $a(C_f)(C_f + I) = I$ so $C_f + I$ is invertible. Hence, the null space of $C_f + I$ has dimension zero, and so $C_f$ cannot have any eigenvectors with eigenvalue one. 
    
    Conversely, if $(x + 1) \mid f(x)$ then we must have $f(1) = f_0 \oplus f_1 \oplus \cdots \oplus f_{d-1} \oplus 1 = 0$ by the factor theorem. Now consider performing row operations on the matrix $C_f + I$. In particular, adding row 1 to row 2, row 2 to row 3, etc., and row $d - 1$ to row $d$, we perform the following transformation:
    $$ \begin{pmatrix}
        1 &   &        &   & f_0\\
        1 & 1 &        &   & f_1\\
          & 1 & \ddots &   & \vdots\\
          &   & \ddots & 1 & f_{d - 2} \\
          &   &        & 1 & f_{d - 1} + 1
    \end{pmatrix} \longrightarrow \begin{pmatrix}
        1 &   &        &   & f_0\\
        0 & 1 &        &   & f_0 \oplus f_1\\
          & 0 & \ddots &   & \vdots\\
          &   & \ddots & 1 & \bigoplus_{i=0}^{d-2} f_{i} \\
          &   &        & 0 & \bigoplus_{i=0}^{d-1} f_{i} + 1
    \end{pmatrix} = \begin{pmatrix}
        1 &   &        &   & f_0\\
        0 & 1 &        &   & f_0 \oplus f_1\\
          & 0 & \ddots &   & \vdots\\
          &   & \ddots & 1 & \bigoplus_{i=0}^{d-2} f_{i} \\
          &   &        & 0 & 0
    \end{pmatrix}$$
    The final matrix clearly has rank $d - 1$, implying the null space of $C_f + I$ has dimension one, and so exactly one eigenvector with eigenvalue one. 
    
    Therefore, each block $C_f$ of $M$ has either zero eigenvectors or exactly one in the case that $f(x) = (x + 1)^d$. By padding each of these eigenvectors to the full width of $M$ and performing a change of basis from the Jordan normal form, we obtain exactly $\lambda = c$ eigenvectors of $M$. Linear independence follows from the fact that they are block diagonal in the normal form basis.
\end{proof}

Combining this with the results of \cite{jiang2023lower}, we can prove that Theorem \ref{thm:cnotcontrolled} is approximately optimal:

\begin{theorem}
    \label{thm:cnotoptimal}
    Theorem \ref{thm:cnotcontrolled} is approximately optimal within a factor of $\frac{3}{2}$.
\end{theorem}

\begin{proof}
    Suppose we are given an $n$-qubit CNOT circuit $C$, then $c(C)$ is a circuit on $n + 1$ qubits. From Lemma \ref{thm:eigenstabs}, the unitary stabilizer nullity \cite[Definition 5]{jiang2023lower} of $c(C)$ can be calculated as
    $$ v(c(C)) = 2(n + 1) - \log_2(2^{2\lambda + 1}) = 2\left(n - \lambda + \frac{1}{2}\right) $$
    and hence from \cite[Theorem 13]{jiang2023lower} the number of Toffoli gates required to synthesize $c(C)$ is bounded below by:
    $$ \operatorname{Toffoli-Count}(c(C)) \geq  \frac{v(c(C))}{v(CCZ)} = \frac{2}{3}\left(n - \lambda + \frac{1}{2}\right) \geq \frac{2}{3}(n - \lambda)$$
    On the other hand Theorem \ref{thm:cnotcontrolled} gives a construction of $c(C)$ using $n - c$ Toffoli gates, and from Lemma \ref{thm:eigenblocks} we have:
    \[ n - c = n - \lambda = \frac{3}{2}\cdot\frac{2}{3}(n - \lambda) \leq \frac{3}{2}\operatorname{Toffoli-Count}(c(C)) \qedhere \]
\end{proof}

We can also use the above bound to provide looser bounds on Theorems \ref{thm:ctrlcnotlowdepth} and \ref{thm:cnotoptimaldepth}, by exploiting the fact that $c \leq \frac{n}{2}$ for any matrix that does not have $C_1$ blocks in its generalized Jordan form. Informally, this says that if the optimal Toffoli count required to control a CNOT circuit is $o(n)$, then there is a basis in which the circuit acts trivially on $\Omega(n)$ qubits.

\begin{theorem}
    \label{thm:depthoptimalapprox}
    The Toffoli-count of the constructions in Theorems \ref{thm:ctrlcnotlowdepth} and \ref{thm:cnotoptimaldepth} are approximately optimal within factors of $6$ and $3$ respectively.
\end{theorem}

\begin{proof}
    Given an $n\times n$ matrix $A$, let $c$ and $c'$ be the number of blocks of the form $C_{(x+1)^d}$ and $C_1$, respectively, in the generalized Jordan form of $A$ (as in Theorem \ref{thm:cnotcontrolled}). Since $(x+1)^0 = 1$, we must have $c \geq c'$. On the other hand, for $d \geq 1$ the size of a block $C_{(x+1)^d}$ is at least $2\times 2$, and so there can be at most $\frac{n - c'}{2}$ of them in the generalized Jordan form of $A$. Hence:
    \begin{align*}
        c &\leq c' + \frac{n - c'}{2} \\
        n - c &\geq n - c' - \frac{n - c'}{2} = \frac{n - c'}{2} \\
        n - c' &\leq 2(n - c)
    \end{align*}
    In Theorem \ref{thm:ctrlcnotlowdepth}, the Toffoli count is exactly $2(n - c') \leq 4(n - c)$, and so by Theorem \ref{thm:cnotoptimal} this is approximately optimal within a factor of $4 \cdot \frac{3}{2} = 6$.

    For Theorem \ref{thm:cnotoptimaldepth}, note that instead of directly synthesizing $c(U)$ where $U$ is the unitary corresponding to parity matrix $A$, we can instead write the generalized Jordan form $A = SJS^{-1}$ and partition $J$ as $J = J_{n-c'} \oplus C_1^{\oplus c'}$ where $J_{n-c'}$ is the $(n-c') \times (n-c')$ part of $J$ that contains no $C_1$ blocks. Since $C_1^{\oplus c'}$ is the $c' \times c'$ identity matrix, we can write $c(U) = U_S [c(U_{J_{n-c'}}) \otimes I]U_S^\dagger$ where $U_S$ and $U_{J_{n-c'}}$ are the CNOT circuits corresponding to $S$ and $J_{n-c'}$. Applying Theorem \ref{thm:cnotoptimaldepth} to just $c(U_{J_{n-c'}})$ has a Toffoli-count of $n - c' \leq 2(n - c)$ while keeping the Toffoli-depth and reaction depth the same. From Theorem \ref{thm:cnotoptimal}, this is approximately optimal within a factor of $2 \cdot \frac{3}{2} = 3$.
\end{proof}

\end{document}